%% file: output.tex
\documentclass[twocolumn]{autart} 

\input{packages}
\input{notations}
\begin{document}

\begin{frontmatter}
  \title{A Forward Reachability Perspective on Control Barrier Functions and Discount Factors in Reachability Analysis}

  \thanks[footnoteinfo]{This work is supported by DARPA Assured Autonomy Grant No. FA8750-18-C-0101, the ONR BRC for Multibody Control Systems No. N000141812214, the NASA ULI Grant No. 80NSSC20M0163, and the NSF
Grant CMMI-1931853. Any opinions, findings, and conclusions expressed in this material are those of the authors and do not necessarily reflect the views of any aforementioned organizations. \\ \hspace*{12pt}* The first two authors contributed equally to this work.}

  \author[UCLA]{Jason J. Choi$^{*}$}\ead{jjhchoi@ucla.edu},
  \author[NC]{Donggun Lee$^{*}$}\ead{donggun\_lee@ncsu.edu},
  \author[UCSD]{Boyang Li},
  \author[MIT]{Jonathan P. How},
  \author[UC]{Koushil Sreenath},
  \author[UCSD]{Sylvia L. Herbert},
  \author[UC]{Claire J. Tomlin}

  \address[UCLA]{University of California, Los Angeles, CA}
  \address[NC]{North Carolina State University, Raleigh, NC}
  \address[UCSD]{University of California, San Diego, CA}
  \address[MIT]{Massachusetts Institute of Technology, MA}
  \address[UC]{University of California, Berkeley, CA}


\begin{abstract}
Control invariant sets are crucial for various methods that aim to design safe control policies for systems whose state constraints must be satisfied over an indefinite time horizon. In this article, we explore the connections among reachability, control invariance, and Control Barrier Functions (CBFs). Unlike prior formulations based on backward reachability concepts, we establish a strong link between these three concepts by examining the inevitable Forward Reachable Tube (FRT), which is the set of states such that every trajectory reaching the FRT must have passed through a given initial set of states. First, our findings show that the inevitable FRT is a robust control invariant set if it has a continuously differentiable boundary. If the boundary is not differentiable, the FRT may lose invariance. We also show that any robust control invariant set including the initial set is a superset of the FRT if the boundary of the invariant set is differentiable. Next, we formulate a differential game between the control and disturbance, where the inevitable FRT is characterized by the zero-superlevel set of the value function. By incorporating a discount factor in the cost function of the game, the barrier constraint of the CBF naturally arises in the Hamilton-Jacobi (HJ) equation and determines the optimal policy. The resulting FRT value function serves as a CBF-like function, and conversely, any valid CBF is also a forward reachability value function. We further prove that any $C^1$ supersolution of the HJ equation for the FRT value functions is a valid CBF and characterizes a robust control invariant set that outer-approximates the FRT. Building on this property, finally, we devise a novel method that learns neural control barrier functions, which learn a control invariant superset of the FRT of a given initial set.
\end{abstract}
\end{frontmatter}


\section{Introduction}

Safety guarantees are essential for control design in many applications. In this article, we focus on safety problems that can be represented by ensuring that system states satisfy specific constraints over an indefinite time horizon. An effective strategy for ensuring state trajectories stay within the desired constraint region involves identifying a subset where the trajectory can remain for an infinite duration. Sets exhibiting these properties are referred to as \textit{control invariant sets} \cite{blanchini1999set}, and are key to various methods for designing safe control policies \cite{Wabersich2023}. The theoretical analysis of control invariance offers valuable insights for the development of safe control policies. 

Control invariant sets can be computed using set-theoretic and optimization-based approaches. For example, convex and scalable methods for computing robust control invariant sets of nonlinear systems have been proposed in~\cite{fiacchini2010computation, schafer2023scalable,comelli2024inner,brown2023computing}. These approaches offer computational advantages and scalability, but typically rely on restrictive assumptions such as convexity, or specific set parameterizations. 

In contrast, approaches that characterize control invariant sets using a scalar function---referred to as the barrier certificate~\cite{Prajna2006}, do not explicitly require such assumptions. This concept has evolved into the notion of a \textit{control barrier function} (CBF)~\cite{Ames2016}, which mandates that the function satisfy a particular differential inequality condition. By enforcing this condition, control policies ensure not just the safety at the boundary of the set, but also a smooth deceleration of trajectories as they approach the boundary. We refer to this condition as the \textit{\cbfconstraint}. Although these concepts provide an effective framework for safety verification and control, designing CBFs remains a nontrivial problem, as the feasibility of the barrier constraint must be ensured everywhere inside the safe set. This challenge has motivated significant recent efforts toward systematic methods for designing or learning CBFs \cite{dawson2023safe}.

This paper focuses on reachability analysis, which is closely related to robust CBFs, namely control invariance under adversarial disturbances~\cite{cardaliaguet1996differential, xue2018reach}. As summarized in Table~\ref{table:comparision_CBF_Reachability}, many reachability-based methods for robust control invariance rely on backward reachability~\cite{Mitchell2005, Fisac2018, akametalu2018minimum, xue2018reach, xue2019robust, choi2021robust}, which characterizes the set of initial states that can remain in a given set despite disturbances. However, these approaches do not directly yield CBFs. Enforcing the barrier condition modifies the associated Hamilton–Jacobi partial differential equation (HJ-PDE), which may result in unbounded or discontinuous value functions and multiple solutions. Consequently, numerical solutions may not necessarily correspond to the robust control invariant set.


Unlike prior reachability methods, our approach leverages forward reachability, which characterizes the states reachable from a given set of initial conditions. The key insight is that this perspective leads to a formulation whose value function naturally satisfies the requirements of control barrier functions, remaining bounded and continuous while enforcing the barrier condition.

To formalize this idea, we extend the notion of the minimal forward reachable tube (FRT)~\cite{mitchell2007comparing} to the inevitable FRT, which represents the set of states that are inevitably reached from the initial set despite worst-case disturbances. For brevity, we use the term “FRT” to denote this inevitable FRT unless specified otherwise. The FRT is characterized through an HJ formulation, which we call the HJ forward reachable tube variational inequality (HJ-FRT-VI).

Next, we show that the FRT is robustly control invariant if its boundary is continuously differentiable. Otherwise, robust control invariance may not hold because the barrier constraint is satisfied almost everywhere on the boundary but not everywhere. Moreover, the viscosity solution of the proposed HJ-FRT-VI may contain non-differentiable points and therefore cannot directly serve as a CBF, even though it satisfies the barrier condition almost everywhere in the state space.

To address this issue, we prove that any $C^1$ continuous supersolution of the proposed HJ-FRT-VI is a valid CBF and characterizes a robust control invariant set that outer-approximates the FRT. Based on this result, we further develop a learning-based numerical method to compute such supersolutions as CBFs.

We highlight the following two main theoretical contributions.
\begin{enumerate}
    \item This paper establishes that CBFs can be systematically constructed using forward reachability, in contrast to existing backward-reachability-based approaches~\cite{Mitchell2005, Fisac2018, akametalu2018minimum, xue2018reach, choi2021robust} (see Table 1). The proposed formulation yields a value function that satisfies the barrier constraint while remaining bounded, continuous, and uniquely defined as the solution of the proposed HJ-FRT-VI.
    
    \item This paper establishes a novel connection between supersolutions of the HJ-FRT-VI and control barrier functions. While supersolutions are a classical concept in the HJ-PDE literature \cite{evans2010partial}, their control-theoretic interpretation in the context of safety and control invariance has not been previously characterized. This paper proves that any $C^1$ supersolution of the proposed HJ-FRT-VI is a valid CBF and characterizes a robust control invariant set that outer-approximates the FRT. This result assigns a concrete safety and invariance meaning to supersolutions and enables the systematic construction of differentiable CBFs. 
\end{enumerate}

The rest of the article is organized as follows. In Section~\ref{sec:probDef}, we review the concepts of control invariance, CBFs, and backward reachability analysis for systems with disturbances. In Section~\ref{sec:FRT_RCI}, we introduce forward reachable tubes and present their relevance to robust control invariance. In Section~\ref{sec:FRT_CBF}, we detail the Hamilton-Jacobi formulation of the FRTs and establish a connection to CBFs. In Section~\ref{sec:frt-cbf-learning}, we propose a method for learning CBFs developed upon the theoretical properties in Sections~\ref{sec:FRT_RCI} and \ref{sec:FRT_CBF}.
Section~\ref{sec:discussion} provides a discussion on how our work is distinguished from the reachability formulations in prior work. We conclude the article with closing remarks in Section~\ref{sec:conclusion}.

\textit{Notation:} $\norm{\cdot}$ indicates the $l^2$ norm in the Euclidean space. For two same dimensional vectors $a$ and $b$,  $a \cdot b$ denotes the inner product. 
For a set $A$, $\interior{A}$, $\partial A$, $A^c$ denote the interior, the boundary, and the complement of $A$, respectively. For a point $x \in \R^n$ and $r > 0$, we define $B_r(x)$ as the hyperball centered at $x$ with radius $r$, $B_r(x):=\{y \in \R^n \;|\; \norm{y-x} \le r\}$. For $\epsilon > 0$ and a set $A$, 
$A+B_\epsilon\coloneqq \bigcup_{x\in A} B_\epsilon(x)$,
and $A-B_\epsilon\coloneqq A \setminus \bigcup_{x\in A^c} B_\epsilon(x)$.

\input{table_comparison}

\section{Background}
\label{sec:probDef}

\subsection{Control Invariance}
\vspace{-0.25em}
We first consider a general nonlinear time-invariant system represented by an ODE
\vspace{-0.25em}
\begin{align}
    \dot{\traj}(t) = f(\traj(t),\ctrl(t))\;\text{for}\;t>0, \qquad \traj(0)=x,
    \label{eq:dynamics_no_d}
\vspace{-0.5em}
\end{align}
where $x\in \R^n$ is an initial state, $\traj:[0,\infty) \rightarrow \R^n$ is the solution to the ODE, and $\ctrl:[0,\infty)\rightarrow U$ is a Lebesgue measurable control signal with $U\subset\R^{m_u}$. We assume that $U$ is compact and use $\cfset$ to denote the set of Lebesgue measurable control signals. 
We also assume that the system \eqref{eq:dynamics_no_d} satisfies the following conditions.
\begin{assumption}[on vector field of \eqref{eq:dynamics_no_d}]~
    \begin{enumerate}
        \item $f:\R^n \times U \rightarrow \R^n ~\textnormal{is uniformly continuous,}$ 
        \item $f(\cdot,u)$ is Lipschitz continuous in $x\in\R^n$ for each $u\in U$, 
        \item $\exists M > 0$ such that $\norm{f(x, u)} \le M$ $\forall x\in \R^n, u \in U$.
    \end{enumerate}
\label{assumption:dynamics_no_d}
\end{assumption}

\noindent Under the above conditions, the solution to the ODE dynamics \eqref{eq:dynamics_no_d} is unique for any $\ctrl\in\cfset$ and initial state $x\in\R^n$. We will call the solution $\traj$ the \textit{(forward) trajectory} from the initial state $x$.

Let $\X\subset \R^n$ be the constraint set, i.e. the set that the system must remain within to maintain safety. The main challenge of finding a control signal $\ctrl\in \cfset$ such that for given $\traj(0) \in \X$, $\traj(t) \in \X$ for all $t \ge 0$ (i.e. $\traj(\cdot)$ remains safe) is that there may be some states in $\X$ from which exiting the set $\X$ is inevitable regardless of the choice of $\ctrl$. An effective way of ruling out these failure states is to consider a subset of $\X$ that is control invariant.




\begin{definition}[(Forward) control invariant \cite{blanchini1999set}]
\label{def:ci}
A set $S \subset \R^n$ in the state space is \textit{(forward) control invariant} under the dynamics \eqref{eq:dynamics_no_d} if for all $x \in S$, there exists a control signal $\ctrl\in\cfset$ such that $\traj(t)\in S$ for all $t\ge0$. We also say that such $\ctrl$ renders the trajectory $\traj$ \textit{forward invariant} in $S$.
\end{definition}

By the above definition, a trajectory starting inside a control invariant set $S$ that is a subset of $\X$ can remain within $S$ for all time, and therefore can stay safe in $\X$. The control invariance can be determined by a geometric relationship between the vector field and the tangent cone of the set:



\begin{lemma}
\label{lemma:tangential_ci_tangent_cone} (Tangential characterization of closed control invariant sets \cite[Theorem 11.3.4]{aubin2011viability}) A closed set $S\subset \R^n$ is (forward) control invariant under the dynamics \eqref{eq:dynamics_no_d} if and only if for all $x\in \partial S$,\vspace{-0.5em}
\begin{equation}
\label{eq:ci_tangent_cone}
    \exists u \in U ~\textnormal{such that}~ f(x, u) \in T_S(x),\vspace{-0.5em}
\end{equation}
where $T_S(x)$ is the tangent cone to $S$.\footnote{Given a closed set $C \subset \R^n$, the (Bouligand's) tangent cone to $C$ at $x \in C$ is defined as
\vspace{-0.5em}
\begin{equation}
T_C(x):=\left\{z \in \R^n \mid \liminf _{\tau \rightarrow 0} \frac{\operatorname{dist}(x+\tau z, C)}{\tau}=0\right\},
\vspace{-0.5em}
\end{equation}
where $\operatorname{dist}(y, C):=\min_{z\in C}{\norm{y -z}}$ \cite{clarke2008nonsmooth}. The tangent cone captures the feasible directions in which one can move from the point $x$ within the set $C$.}
\end{lemma}

We consider a corollary of the lemma in a special case when the set $S$ has a differentiable boundary (Condition~\ref{assumption:S}), by introducing a scalar function $\cbf:\R^n \rightarrow \R$ that satisfies Condition~\ref{assumption:hS}. 


\begin{condition}[Differentiability of the boundary]
\label{assumption:S} $C$ is a closed set whose interior is non-empty, and whose boundary, $\partial C$, \textit{is continuously differentiable}.\footnote{For each point $x\in\partial C$, there exists $r>0$ and a $C^1$ function $\eta:\R^{n-1}\rightarrow \R$ such that $C \cap B_r(x) = \{x\in B_r(x) ~|~ x_n \ge \eta(x_1,...,x_{n-1}) \}$, where relabeling and reorienting the coordinates axes are allowed \cite{evans2010partial}.}
\end{condition}

\begin{condition}
Given a closed set $C$, $h_C:\R^n\to\mathbb{R}$ is a function whose zero-superlevel set is $C$, i.e. $C = \{x\in\mathbb{R}^{n} ~|~ h_C(x) \geq 0\}$, and satisfies the following conditions:
\begin{enumerate}
    \item \;
    \vspace{-2em}
    \begin{align}
        \interior{C} & = \{x\in\mathbb{R}^{n} ~|~ h_C(x) > 0\}, \nonumber \\
        \partial C & = \{x\in\mathbb{R}^{n} ~|~ h_C(x) = 0\}. 
    \label{eq:cond_for_S1}
    \vspace{-1em}
    \end{align}        
    \item (Differentiability and boundedness) $h_C$ is uniformly continuously differentiable and both upper and lower bounded.
    \item  (Regularity) $\exists {\epsilon} > 0$ such that
    \vspace{-0.5em}
        \begin{align}
            \frac{\partial h_C}{\partial x}(x) \neq 0 \quad \forall x\in\partial C + B_{\epsilon}.
            \label{eq:cond_for_S2}  
    \vspace{-0.5em}
        \end{align}
        \end{enumerate}
\label{assumption:hS}
\end{condition}

\begin{proposition}
\label{prop:S-regularity} For any $h_C$ satisfying Condition~\ref{assumption:hS}, its zero-superlevel set $C$ satisfies Condition~\ref{assumption:S} \cite[Theorem 9.28]{rudin2021principles}, and for any set $C$ satisfying Condition~\ref{assumption:S}, $h_C$ satisfying Condition~\ref{assumption:hS} always exists \cite[Theorem 2.1]{lieberman1985regularized}.
\end{proposition}

\noindent Given these conditions, Lemma~\ref{lemma:tangential_ci_tangent_cone} results in the following corollary: 
\begin{lemma}
\label{lemma:tangential_ci} $S$ satisfying Condition~\ref{assumption:S} is (forward) control invariant under the dynamics \eqref{eq:dynamics_no_d} if and only if for all $x\in \partial S$,
\vspace{-0.5em}
\begin{equation}
\label{eq:ci_tangent}
    \exists u \in U ~\textnormal{such that}~ \dcbf(x) \cdot f(x, u) \ge 0,
    \vspace{-0.5em}
\end{equation}
for $h_S$ satisfying Condition~\ref{assumption:hS}.
\end{lemma}

\begin{proof} This is a corollary of Lemma~\ref{lemma:tangential_ci_tangent_cone} by noticing that for $x\in \partial S$, 
\vspace{-0.75em}
\begin{equation}
    T_S(x) = \left\{ z \in \R^n | \dcbf \cdot z \ge 0 \right\}.
    \label{eq:TS}
\vspace{-0.5em}
\end{equation}
\vspace{-1em}
\end{proof}
\vspace{-0.5em}

Lemma~\ref{lemma:tangential_ci} is known as Nagumo's theorem for autonomous systems \cite{nagumo1942lage}. For a given $S$ that is control invariant, \eqref{eq:ci_tangent} holds for any $\cbf$ satisfying Condition~\ref{assumption:hS}. The specific choice of $\cbf$ does not affect condition \eqref{eq:ci_tangent}. 



Next, we extend the concept to systems with disturbance. There exist various formulations of robustness with respect to disturbances or uncertainties in dynamics \cite{blanchini1999set, jankovic2018robust, kolathaya2018input, krstic2021inverse}. In this paper, we employ the differential game that interprets the disturbance as an adversary playing against the control input \cite{evans1984differential}, as commonly done in the Hamilton-Jacobi analysis \cite{cardaliaguet1996differential, Fisac2018, xue2018reach}. For this, we consider the system dynamics
\vspace{-0.5em}
\begin{align}
    \dot{\traj}(t) = f(\traj(t),\ctrl(t),\dstb(t))\;\text{for}\;t>0,\qquad \traj(0)=x,
    \label{eq:dynamics}
\vspace{-0.5em}
\end{align}
where $\dstb:[0,\infty)\rightarrow D$ is a Lebesgue measurable disturbance signal and $D\subset\R^{m_d}$ is a compact set. Control systems \eqref{eq:dynamics_no_d} can be regarded as a special case when the disturbance set $D$ in \eqref{eq:dynamics} is set to a singleton (e.g. $D = \{0\}$). We use $\dfset$ to denote the set of Lebesgue measurable disturbance signals. We assume conditions on the dynamics, similar to Assumption~\ref{assumption:dynamics_no_d}:
\begin{assumption}[on vector field of \eqref{eq:dynamics}]~
    \begin{enumerate}
        \item $f:\R^n \times U \times D \rightarrow \R^n ~\textnormal{is uniformly continuous,}$ 
        \item $f(\cdot,u,d)$ is Lipschitz continuous in $x\in\R^n$ for each $(u,d)\in U\times D$, 
        \item $\exists M > 0$ such that $\norm{f(x, u, d)} \le M$ $\forall x\in \R^n, (u, d) \in U \times D$,
    \end{enumerate}
\label{assumption:dynamics}
\vspace{-0.5em}
\end{assumption}

\noindent so that under the above conditions, the solution to the ODE \eqref{eq:dynamics} is unique for any pair of $(\ctrl,\dstb)\in\cfset\times\dfset$ and initial state $x\in\R^n$ \cite{evans1984differential}.

To ensure safety under the most adversarial disturbance, we assume that the disturbance can use the control signal's current and previous information, whereas the control is not aware of the current disturbance input, by considering the notion of the \textit{non-anticipative strategies} \cite{evans1984differential}:
\vspace{-0.75em}
\begin{align}
\label{eq:nonanticipative}
\begin{split}
    \Strategy\coloneqq \{& \strategy:\cfset\rightarrow\dfset~|~\forall s\in[0,\infty) \text{ and } \ctrl,\bar\ctrl\in\cfset, \\
    &\text{if } \ctrl(\tau)=\bar\ctrl(\tau) \text{ a.e. }\tau\in[0,s],\; \\
    &\text{then } \strategy[\ctrl](\tau)=\strategy[\bar\ctrl](\tau) \text{ a.e. }\tau \in[0,s] \}.
\end{split}
\vspace{-0.75em}
\end{align}
Using the notion of non-anticipative strategies, we define the robust control invariant set under the dynamics \eqref{eq:dynamics}.

\begin{definition}[Robustly (forward) control invariant \cite{cardaliaguet1996differential, xue2018reach}]
\label{def:rci}
A set $S \subset \R^n$ is \textit{robustly (forward) control invariant} (under the dynamics \eqref{eq:dynamics}) if, for all $x \in S$, $\strategy\in\Strategy$, for any $\epsilon > 0$ and $T > 0$, there exists a control signal $\ctrl(\cdot)\in\cfset$ such that $\traj(t)\in S + B_\epsilon$ for all $t\in[0, T]$.
\end{definition}
If $S$ is an open set, the notion of $\epsilon$ and $T$ can be dropped. 
However, at the boundary of a closed set $S$, the disturbance can react to the current control input to drive the system outside of $S$. Thus, $\traj$ might not stay in $S$ for all time although the trajectory $\traj$ will stay in $S+B_\epsilon$ for any small $\epsilon$. 
An example that elucidates the necessity of $\epsilon$ and $T$ is provided in \cite{cardaliaguet1996differential}.

Similar to Lemma~\ref{lemma:tangential_ci}, robust control invariant sets can be verified by examining the vector field of the dynamics at the boundary of the sets: 

\begin{lemma} \label{lemma:tangential} 
(Tangential characterization of robust control invariant sets) $S$ satisfying Condition~\ref{assumption:S} is robustly (forward) control invariant under the dynamics \eqref{eq:dynamics} if and only if for all $x\in \partial S$,
\vspace{-0.5em}
\begin{equation}
\label{eq:rci_tangent}
    \exists u \in U ~\textnormal{such that}~ \dcbf(x) \cdot f(x, u, d) \ge 0 \; \forall d\in D,
\vspace{-0.5em}
\end{equation}
for $h_S$ satisfying Condition~\ref{assumption:hS}.
\end{lemma}
\vspace{-0.5em}
\begin{proof} This results from \cite[Theorem 2.3]{cardaliaguet1996differential}, and \eqref{eq:TS}.
\end{proof}

\subsection{Control Barrier Functions}
Lemma~\ref{lemma:tangential_ci} implies that a safe control input on the boundary of the set $S$ can render the trajectory $\traj$ forward invariant in $S$. Control barrier functions (CBFs), first introduced in \cite{Ames2016}, additionally impose conditions on the input when the trajectory is strictly inside the set before reaching the boundary. Here, we extend its definition to robust CBF for the dynamics with disturbance \eqref{eq:dynamics}.

\begin{definition}[Robust Control Barrier Function] \label{def:rcbf}
A function $\cbf:\R^n\rightarrow \R$ that satisfies Condition~\ref{assumption:hS} for a closed set $S$ is a \textit{robust CBF} for the dynamics \eqref{eq:dynamics} if there exists an extended class $\mathcal{K}$ function $\alpha$ such that, for all $x\in S$,
\vspace{-0.75em}
\begin{align}
    \max_{u\in U}\min_{d \in D} \dcbf(x)\cdot f(x,u,d) +\alpha\left(\cbf(x)\right) \geq 0.
    \label{eq:rCBF_ctrl_condition_original}
\vspace{-0.5em}
\end{align}
Here, $\alpha:\R\rightarrow\R$ is an extended class $\mathcal{K}$ function if it is continuous and strictly increasing and satisfies $\alpha(0)=0$. We say the \textit{\cbfconstraint~is feasible at} $x$ if condition \eqref{eq:rCBF_ctrl_condition_original} holds for $x$. 
\end{definition}

This paper considers a particular class $\mathcal{K}$ function $\alpha(y)=\discount y$ for a constant $\discount>0$, as in \cite{ogren2006autonomous, nguyen2016exponential},
\vspace{-0.5em}
\begin{align}
    \max_{u\in U}\min_{d \in D} \dcbf(x)\cdot f(x,u, d) +\discount \cbf(x) \geq 0.
    \label{eq:CBF_ctrl_condition}
\vspace{-0.5em}
\end{align}
This is the most common choice of class $\mathcal{K}$ function used in the CBF literature, and enables us to make a connection between CBFs and reachability value functions where $\discount$ will play the role of a discount factor in the reachability formulation. Intuitively, \eqref{eq:CBF_ctrl_condition} ensures that $\cbf(\traj(t))$ does not decay faster than the exponentially decaying
curve $\dot{h}_{S} (\traj(t)) = - \discount  \cbf (\traj(t))$. This induces the braking mechanism to any trajectory $\traj$ approaching the boundary of $S$. At $x \in \partial S $ where $\cbf (x) = 0$, \eqref{eq:CBF_ctrl_condition} implies \eqref{eq:rci_tangent}; thus, the existence of the robust CBF $\cbf$ is a sufficient condition for $S$ being robust control invariant: 

\begin{proposition} If a robust CBF $\cbf$ exists, $S$ is robustly control invariant. 
\label{prop:rcbf}
\end{proposition}

\begin{remark} \label{rm:cbf-regularity}
Since a robust CBF $h_S$ must satisfy Condition~\ref{assumption:hS} (originally introduced in \cite{Ames2016}), the invariant set $S$ must satisfy Condition~\ref{assumption:S} by Proposition \ref{prop:S-regularity}.
\end{remark}


\subsection{Backward reachability for control invariance}
\label{sec:brt}

We now introduce how backward reachability analysis is used to construct a maximal robust control invariant set contained in the constraint set $X\subset\mathbb{R}^n$. We consider the following definition of the backward reachable tube (BRT):

\begin{definition}[Inevitable BRT] 
\label{def:brt}
For a given terminal set $C \subset \R^n$ which is an open set, we define the (infinite-horizon inevitable) BRT of $C$ as the following set.
\vspace{-0.5em}
\begin{align}
     & \brt{C} \coloneqq \Big\{ x\!\in\! \R^n ~\Big|~\exists \strategy\in\Strategy,T>0 \text{ s.t.} \forall \ctrl \in \cfset,\;\;\;\;\;\; \nonumber \\
    &\; \exists t\!\in\![0, T] \text{ s.t. } \traj(t)\!\in\!C, \text{ where } \traj \text{ solves } \eqref{eq:dynamics}.\!\Big\}
\label{eq:FRT}
\end{align}
\end{definition}
\vspace{-0.75em}
Consider the inevitable BRT of the failure region $X^c$. Then, the following holds:

\begin{proposition}
\label{prop:max-rci}
For a closed set $X$, $\{\brt{X^c}\}^c$ is the maximal robust control invariant set contained in $X$ \cite{cardaliaguet1996differential, Fisac2018}.
\end{proposition}

This maximal robust control invariant set is also known as the (robust) \textit{viability kernel} of the constraint set $X$ \cite{aubin2011viability}.

The $\brt{X^c}$ and its complement can be computed by considering the following BRT value function:
\begin{definition}[BRT value function with discount factor \cite{akametalu2018minimum, xue2018reach}]\;
\vspace{-0.75em}
\begin{equation}
\label{eq:brt_v_discounted}
V(x) := \underset{{\strategy \in \Strategy}}{\inf} \; \underset{\ctrl \in \cfset}{\sup}\; \underset{t\in[0,\infty)}{\inf}  e^{-\lambda t} h_{X}(\traj(t)),
\end{equation}   
\vspace{-2em}
\end{definition}
where $h_{X}$ is a function that satisfies $X = \{x\in\mathbb{R}^{n} ~|~ h_X (x) \ge 0\}$. The control tries to maximize the minimum value of $h_X$ discounted in time. The value function is the viscosity solution \cite{bardi1997optimal} to the following Hamilton-Jacobi backward reachable tube variational inequality (HJ-BRT-VI):
\vspace{-0.75em}
\begin{align}
        \displaystyle 0=\min\Big\{ \
        & h_X(x) - V(x), \label{eq:thm_HJPDE_vanilla} \\
        & \max_u \min_d \frac{\partial V}{\partial x}\cdot f(x,u,d) - \lambda V(x)\Big\} \nonumber \vspace{-0.75em}
\end{align}

\begin{figure*}
\centering
\includegraphics[width=\textwidth]{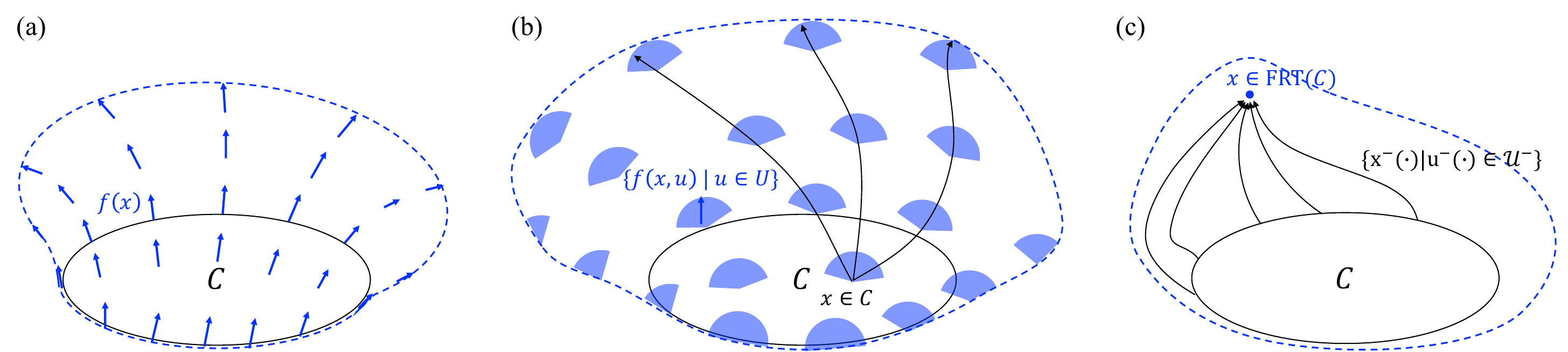}
\vspace{-1em}
\caption{(a) For an autonomous system $\dot{\traj}(t) = f(\traj(t))$, the forward reachable tube (FRT) of the initial set $C$ is the union of the forward trajectories starting from $C$. (b) For control systems, the viable (maximal) FRT is the collection of all possible trajectories that depart from $C$. (c) On the other hand, the inevitable (minimal) FRT, the main focus of this study, is a collection of states such that every trajectory reaching it must have passed through $C$ at some point in the past.}
\label{fig:frt_defs}
\vspace{-.5em}
\end{figure*}

\vspace{-1em}
We now discuss why this value function has little relevance to robust CBFs in Definition \ref{def:rcbf}, and is not ideal for safety control although it is useful to verify the maximal control invariant set. When $\lambda = \gamma > 0$, because of the discount term applied to $h_X$ indefinitely in time, for all states within the invariant set, $\inf_{t\in[0,\infty)}  e^{-\lambda t} h_{X}(\traj(t))$ flattens to zero, forming the maximal control invariant set as $\{x\in\mathbb{R}^{n} ~|~ V (x) = 0\}$ \cite{xue2018reach} (Table \ref{table:comparision_CBF_Reachability}, first row). Within the invariant set, the value function has zero gradient, making it not usable for constraining the control input as in the barrier constraint \eqref{eq:CBF_ctrl_condition}. In fact, as seen in \eqref{eq:thm_HJPDE_vanilla}, the value function satisfies the inequality condition \vspace{-0.5em}
\[
\max_u \min_d \frac{\partial V}{\partial x}\cdot f(x,u,d) - \lambda V(x) \ge 0.
\vspace{-0.5em}
\]
To make this condition match with the barrier constraint, we need to consider $\lambda = -\gamma < 0$, resulting in an anti-discount (or an exploding factor) to the cost function. 
The work in \cite{choi2021robust} presents this formulation for the finite-horizon BRT, which is called the control barrier-value function (CBVF). However, its extension to infinite horizon might lead to an unbounded and discontinuous value function and non-unique solutions to the corresponding HJ-PDE (Table \ref{table:comparision_CBF_Reachability}, second row). Similarly, the value function without any discount ($\lambda = 0$) can lead to discontinuous and non-unique solutions to the corresponding HJ-PDE (Table \ref{table:comparision_CBF_Reachability}, third row).

Finally, in general, the maximal control invariant set contained in the desired safety constraint set $\X$ might have a non-differentiable boundary \cite{aubin2011viability}. However, the differentiability of the boundary of the zero-superlevel set of the CBF is a feature derived from its definition, as discussed in Remark~\ref{rm:cbf-regularity}, which prohibits designing a CBF based on the maximal control invariant set.

These limitations of backward reachability approaches in designing CBFs motivate us to investigate the forward reachability as an alternative, which we discuss next.

\section{Forward Reachability Analysis}
\label{sec:FRT_RCI}

\vspace{-0.5em}

We now introduce forward reachability analysis to control invariant sets. We first provide background on the forward reachability of a set.

\vspace{-0.25em}

\subsection{Forward Reachability}
\label{sec:FRT}

\vspace{-0.25em}

Forward reachability analyzes a set's evolution in the future. The forward reachable tube (FRT) of a set, roughly speaking, encompasses states that are reached by trajectories that depart from the initial set. This concept is illustrated in Figure \ref{fig:frt_defs}. For autonomous systems (e.g. $\dot{\traj}(t) = f(\traj(t))$), the forward evolution of a set is uniquely determined (Fig~\ref{fig:frt_defs}~(a)). However, for systems with control and/or disturbance inputs, the trajectory, and thus the forward reachable tube, can be determined in various ways according to control and disturbance.

Consider the dynamics without disturbance, satisfying \eqref{eq:dynamics_no_d}. At one extreme, the control can use its best effort to get further away from the original set $C$, and at the other extreme, the control works to stay as close to $C$ as possible. The former would cause the FRT to expand maximally covering all the states such that from $C$, reaching them is viable (Fig~\ref{fig:frt_defs}~(b)), and the latter would induce the FRT to grow minimally, encapsulating only the states which inevitably must have evolved from $C$ (Fig~\ref{fig:frt_defs}~(c)). From this intuition, we are able to define the viable (maximal) and inevitable (minimal) FRTs of a set $C$. Here, we only define the inevitable FRT which is the focus of this study, and readers are referred to \cite{mitchell2007comparing} for the definition of the viable FRT. Applications of forward reachability in the safe control and verification literature primarily focused on determining the viable FRTs, and checking whether this set intersects with the unsafe set \cite{althoff2014online, kousik2020bridging, wetzlinger2022fully}, with limited interest in the utility of inevitable FRTs \cite{mitchell2007comparing}.


To introduce the formal definition of the inevitable FRT, we use a separate notation for the solution of the ODE whose \textit{terminal} state is specified as $x$ (as opposed to initial states being specified as $x$ in \eqref{eq:dynamics}):
\vspace{-0.5em}
\begin{equation}
\resizebox{0.87\hsize}{!}{$\displaystyle
    \dot{\traj}^- (t) = f(\BWtraj(t),\BWctrl(t),\BWstrategy[\BWctrl](t)), \;t<0,\;\; \BWtraj(0) = x,$}
    \label{eq:backward_dynamics}
\vspace{-1em}
\end{equation}

where $\BWctrl:(-\infty,0]\rightarrow\cset$ is an element of $\BWcfset$, a set of measurable backward control signals. Evaluating whether $\BWtraj$ reaches the set $C$ when time is considered to flow backward will tell us whether $x$ belongs to the forward reachable tube of $C$. With $\BWdfset$ denoting a set of measurable backward disturbance signals $\BWdstb:(-\infty,0]\rightarrow\dset$, $\BWstrategy$ is a non-anticipative strategy for the disturbance \textit{backward} in time.
The separate notation is necessary due to the causality of the disturbance strategy being reversed in time. 
\begin{definition}[Inevitable FRT] 
\label{def:frt}
For a given initial set $C \subset \R^n$ which is an open set, we define the (infinite-horizon inevitable) FRT of $C$ as the following set.
\vspace{-0.5em}
\begin{align}
     & \frt{C} \coloneqq \Big\{ x\!\in\! \R^n ~\Big|~\exists \BWstrategy\in\BWStrategy,T>0 \text{ s.t.} \forall \ctrl \in \BWcfset,\;\;\;\;\;\; \nonumber \\
    &\; \exists t\!\in\![-T,0] \text{ s.t. } \BWtraj\!(t)\!\in\!C, \text{ where } \BWtraj \text{solves } \eqref{eq:backward_dynamics}\!\Big\}
\label{eq:FRT}
\end{align}
\end{definition}

\vspace{-1em}

Note that ~$\frt{\cdot}$ can be interpreted as a set mapping, $\text{FRT}:2^{\R^n}\rightarrow2^{\R^n}$. In words, $\minfrt{C}$ is a collection of states such that \textit{every} trajectory reaching it forward in time must have passed through $C$ at some point in the past. The FRT is shaped by the control aiming to restrain the growth of the FRT, whereas the disturbance is assumed to act adversarially and attempts to grow the FRT. For control systems, the inevitable FRT is minimal \cite{mitchell2007comparing} since it excludes any state that can be reached by a trajectory that does not evolve from $C$.

\begin{figure*}[h]
\centering
\vspace{-0.5em}
\includegraphics[width=\textwidth]{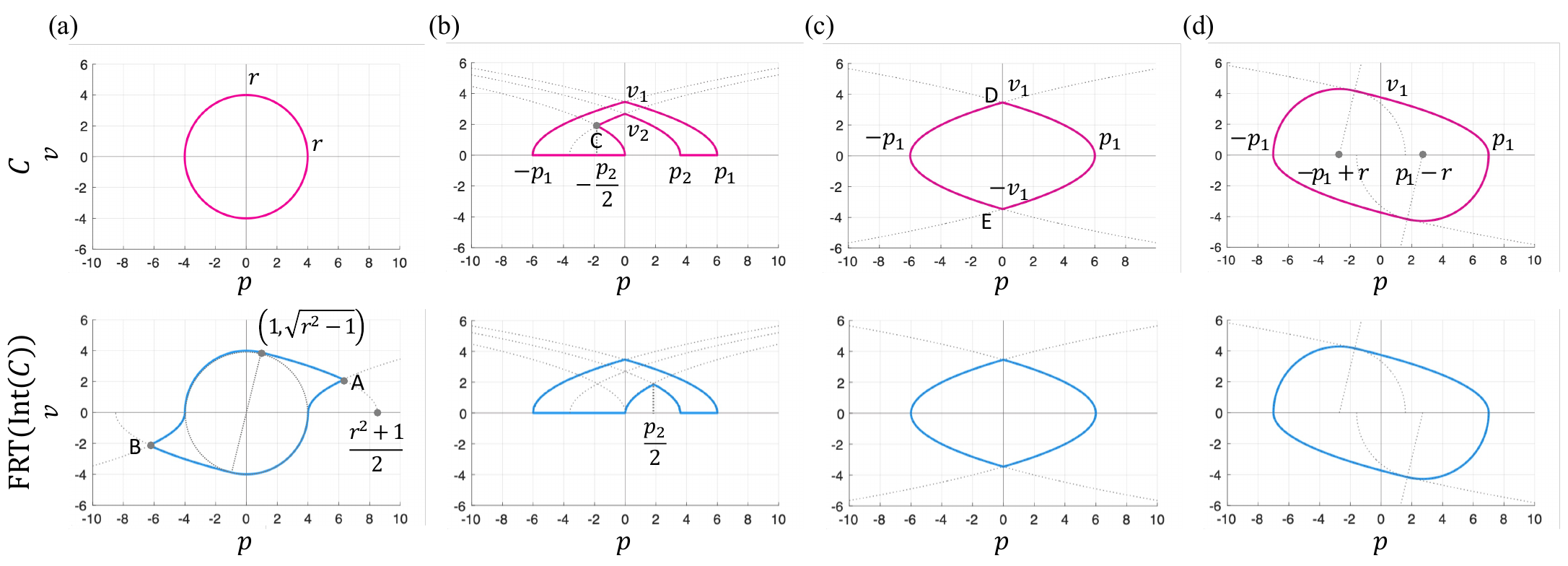}
\vspace{-2em}
 \caption{\textcolor{black}{Forward} reachable tubes under double integrator dynamics for various shapes of the initial set $C$. In the first row, $C$ is visualized as the interior of the pink level curve. The interior of the blue level curve in the second row is $\frt{\interior{C}}$ for each case. (a) Smooth $C_a$ that is not control invariant, resulting in $\frt{\interior{C_a}}\neq \interior{C_a}$. The FRT is still not control invariant. (b, c) Nonsmooth sets $C_b$, $C_c$ that are control invariant. In case (b), $\frt{\interior{C_b}}\neq \interior{C_b}$; this shows that Condition~\ref{assumption:S} is required for Corollary~\ref{cor:frt_corollary} to hold. In case (c), $\frt{\interior{C_c}}=\interior{C_c}$.  (d) Smooth control invariant $C_d$ that results in $\frt{\interior{C_d}} = \interior{C_d}$ according to Corollary~\ref{cor:frt_corollary}.}
\label{fig:double_integrator}
\vspace{-0.5em}
\end{figure*}

\subsection{Forward Reachable Tubes of Control Invariant Sets}
\label{sec:FRT_CtrlInv}
\vspace{-0.5em}

We now investigate the relationship between the inevitable FRTs and control invariance. Since the inevitable BRT of $X^c$ led to the discovery of the maximal control invariant set in the constraint set $X$ (Proposition~\ref{prop:max-rci}), one might hypothesize that the inevitable FRT may lead to the minimal control invariant set containing the initial set $C$. However, we show that this is only true when the boundary of the FRT is differentiable:

\begin{theorem}\label{thm:frt}    
    For any given set $C$, if $\frt{\interior{C}}$ satisfies Condition~\ref{assumption:S}, $\frt{\interior{C}}$ is (the minimal) robust control invariant set under~\eqref{eq:dynamics} containing $C$.
\end{theorem}
\vspace{-0.75em}
\begin{proof}
See Appendix~\ref{appendix:thm_frt}.
\end{proof}

Condition~\ref{assumption:S} is crucial for the theorem's proof, \finalrev{which establishes the equivalence between $f(x, \bar \pi(x), d) \in T_S(x)$ and $- f(x, \bar \pi(x), d) \in T_{\interior{S}^c}(x)$  for all $d\in\dset$ (Lemma~\ref{lemma:positive_negative}), where $S=\frt{\interior{C}}$.} At the non-differentiable boundary of a set $S$, this is not necessarily true. In this case, a forward trajectory can be inevitably ``leaked'' from the interior of the FRT through the point at the boundary where it is non-differentiable. An example of this incident is introduced in Section \ref{subsec:di1} (Figure~\ref{fig:double_integrator}a).

In addition, the following corollary shows that the interior of any control 
invariant set with a differentiable boundary is a fixed point of the $\frt{\cdot}$ set operator.

\begin{corollary}\label{cor:frt_corollary}        
    A set $S$ satisfying Condition~\ref{assumption:S} is robustly control invariant under \eqref{eq:dynamics} if and only if $\frt{\interior{S}}=\interior{S}$.
\end{corollary}
\vspace{-0.75em}
\begin{proof}
    See Appendix \ref{appendix:proof-frt-fixed-point}.
\end{proof}
\vspace{-0.75em}
As noted in Remark \ref{rm:cbf-regularity}, since control invariant sets associated with any CBFs have to satisfy Condition~\ref{assumption:S}, Corollary~\ref{cor:frt_corollary} implies that this fixed point property is a key feature of a valid CBF design.

Finally, in the case when the FRT or the initial set $C$ does not satisfy Condition~\ref{assumption:S}, the FRT still provides valuable information in characterizing an invariant set that contains $C$: 


\begin{theorem}\label{cor:overapproximation}
For any given set $C$, the minimal robust control invariant set $S$ which includes $C$ and satisfies Condition~\ref{assumption:S} is always a superset of $\frt{\interior{C}}$.
\end{theorem}
\vspace{-0.75em}
\begin{proof}
    See Appendix \ref{appendix:proof-frt-overapproximation}.
\end{proof}
\vspace{-0.75em}
Note that again, Condition~\ref{assumption:S} is necessary for the invariant set $S$ to be the superset of $\frt{\interior{C}}$. A counterexample when Condition~\ref{assumption:S} does not hold is explained in Section \ref{subsec:di1} (Figure~\ref{fig:double_integrator}b).

This theorem implies that the FRT can be used to lower-bound when finding a control invariant set that contains $C$. We extend this property further in Sections \ref{sec:HJ_FRT_CBF} and \ref{sec:frt-cbf-learning} to devise a CBF design method whose control invariant set overapproximates the FRT. 


\subsection{Example: Double Integrator}
\label{subsec:di1}

We introduce examples of four different initial sets (Figure~\ref{fig:double_integrator} first row) and their FRTs (Figure~\ref{fig:double_integrator} second row) to illustrate the importance of the differentiability of its boundary for control invariance. These examples are based on a simple double integrator system defined by $\dot{p}  = v, \dot{v}  = u$, with state $x=[p \; v]^T$, and control input $u$, with control bound $u \in [-1, 1]$.
Note that curves $p=c\pm\frac{1}{2}v^2$ characterize trajectories with the saturated input $u=\pm1$, $\ddot{p}=\pm 1$, that decelerate or accelerate until $v=0$.
The four specific $C$ sets are: 
\begin{enumerate}[label=(\alph*)]
\item $C_a$ is defined as the circular region with radius $r$ centered at the origin. This set $C_a$ satisfies Condition~\ref{assumption:S} but is not control invariant.
\item $C_b$ is formed by five curves, $p = -p_1+\frac{1}{2}v^2$, $p = p_1-\frac{1}{2}v^2$, $p =p_2-\frac{1}{2}v^2$, $p = -p_2+\frac{1}{2}v^2$, $p=-\frac{1}{2}v^2$, and the $p$-axis, as shown in Figure~\ref{fig:double_integrator}(b). This set does not satisfy Condition~\ref{assumption:S}, but is control invariant.
\item $C_c$ is formed by two curves, $p = -p_1+\frac{1}{2}v^2$, and $p = p_1-\frac{1}{2}v^2$. This set also does not satisfy Condition~\ref{assumption:S} but is control invariant.
\item $C_d$ is formed by two curves, $p = -p_1+\frac{1}{2}v^2$ and $p = p_1-\frac{1}{2}v^2$ ($p_1 > 1$), and two arcs whose radius is $r=-1+2\sqrt{p_1}$ that are tangential to the curves whose centers are positioned at $(-p_1 + r, 0), (p_1 -r, 0)$, respectively. This set satisfies Condition~\ref{assumption:S} and is also control invariant.
\end{enumerate}

\begin{figure}[t]
\centering
\includegraphics[width=0.75\columnwidth]{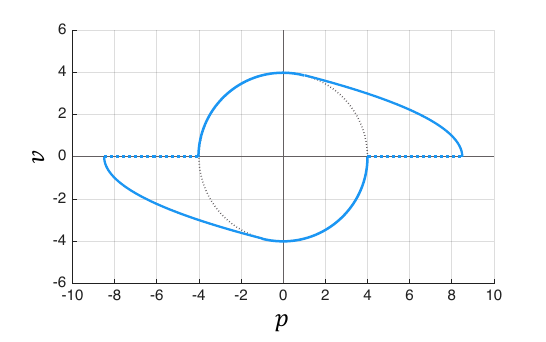}
\vspace{-1em}
\caption{Minimal control invariant set that contains $C_a$ in Figure \ref{fig:double_integrator}a, which is a superset of $\frt{\interior{C_a}}$.}
\label{fig:case_a_inv}
\end{figure}

The first case demonstrates that the FRT can be a strict superset of $\interior{C_a}$ when $C_a$ is not control invariant. The resulting $\frt{\interior{C_a}}$ is still not control invariant since the trajectory is bound to escape the set at points A and B. The actual minimal control invariant set that contains $C_a$ is visualized in Figure ~\ref{fig:case_a_inv}. This example reveals a challenge in constructing a control invariant set with exact FRT, when the initial set is not control invariant. Thus, in the next section, we focus on utilizing Theorem~\ref{cor:overapproximation} to find an outer-approximated FRT that is control invariant.

In the second and third cases, the control invariance of $C_b$ and $C_c$ can be checked analytically. The second case shows that the FRT can be a strict superset of $\interior{C_b}$ if Condition~\ref{assumption:S} is not met, even though $C_b$ is control invariant, showing the necessity of Condition~\ref{assumption:S} in Corollary~\ref{cor:frt_corollary}. Note that point C is where $f(x, \bar \pi(x))\!\in\!T_C(x)$ holds but $- f(x, \bar \pi(x))\!\in\!T_{\interior{C}^c}(x)$ does not hold. Moreover, since $C_b$ itself is the minimal control invariant set containing $C_b$, Theorem~\ref{cor:overapproximation} does not hold due to the absence of Condition~\ref{assumption:S}.

In the third case, since $- f(x, \bar \pi(x))\!\in\!T_{\interior{C}^c}(x)$ holds at both points D and E where $\partial C_c$ is not smooth, $\frt{\interior{C_c}}=\interior{C_c}$. Finally, the set $C_d$ in the last case satisfies Condition~\ref{assumption:S} and is also control invariant. Thus, according to Corollary~\ref{cor:frt_corollary}, $\frt{\interior{C_d}}$ remains the same as $\interior{C_d}$.


\vspace{-0.5em}

\section{FRT value function and CBF}
\label{sec:FRT_CBF}

\vspace{-0.5em}

In this section, by taking the Hamilton-Jacobi approach to the forward reachability problem, we pose the computation of FRT as a differential game. The FRT is characterized by the value function proposed in Section \ref{sec:HJ_FRT}. This value function is the unique solution to the HJ-PDE proposed in Section \ref{subsec:frt-hj}. In Section \ref{sec:HJ_FRT_CBF}, we establish a connection between the proposed value function and the CBFs, by taking the supersolution of the HJ-PDE. This provides an interpretation of any valid CBF as a forward reachability value function, and also a constructive way to design CBF as a supersolution of the FRT value function. 

\subsection{FRT Value Function with Discount Factor}
\label{sec:HJ_FRT}

\vspace{-0.5em}

By noting that $h_C$ satisfying Condition~\ref{assumption:hS} serves as a distance-like metric to the boundary of $C$ and its sign serves as an indicator of the inclusion in $C$, we can rewrite the definition of the FRT in \eqref{eq:FRT} as follows:
\vspace{-0.5em}
\begin{align*}
         & \minfrt{\interior{C}} \\ 
         & = \!\big\{ x ~|~ \exists \BWstrategy\in\BWStrategy, \forall \BWctrl \!\in \!\BWcfset, \!\sup_{t\in(-\infty, 0]} h_C(\BWtraj(t)) \!>\!0 \big\} \\
    & = \!\big\{ x ~|~ \sup_{\BWstrategy\in \BWStrategy} \inf_{\BWctrl\in\BWcfset} \sup_{t\in(-\infty, 0]} h_C(\BWtraj(t)) > 0 \big\}. \vspace{-0.75em}
\end{align*}
Since rescaling $h_C(\BWtraj(t))$ with a positive constant at any time $t$ does not change its sign, the following holds:
\vspace{-0.5em}
\[
\resizebox{\hsize}{!}{$\displaystyle \minfrt{\interior{C}}\! = \Big\{ x ~| \sup_{\BWstrategy\in \BWStrategy} \inf_{\BWctrl\in\BWcfset} \sup_{t\in(-\infty, 0]} e^{\discount t} h_C(\BWtraj(t)) > 0 \Big\},$}\vspace{-1em}
\]         
where $\gamma \!>\!0$ and at each time $t\!\in\!(-\infty, 0]$, $h_C(\BWtraj(t))$ is rescaled by $e^{\discount t}$.
Thus, by defining the FRT value function of $C$, $\VV:\!\R^n\!\rightarrow\!\R$, as 
\vspace{-0.5em}
\begin{equation}
    \label{eq:Def_FRT_valueFunction}
    \VV(x) \coloneqq \sup_{\BWstrategy\in \BWStrategy}\inf_{\BWctrl\in\BWcfset} \JJ(x, \BWctrl, \BWstrategy)
\vspace{-1em}
\end{equation}
with the cost functional $\JJ\!:\!\R^n \!\times\! \BWcfset \!\times \BWStrategy\!\rightarrow\!\R$ defined as
\vspace{-0.5em}
\begin{equation}
    \label{eq:frt_cost}
    \JJ(x, \BWctrl, \BWstrategy) = \sup_{t\in(-\infty,0]}  e^{\discount t}h_C(\BWtraj(t)), \vspace{-1em}
\end{equation}
where $\BWtraj$ solves \eqref{eq:backward_dynamics} and $x$ is the terminal state of $\BWtraj$.

The value function $\VV$ captures a differential game between the control and the disturbance, wherein the optimal control signal of this game is verifying the existence of a trajectory that reaches $x$ without passing through $\interior{C}$ in the past under the worst-case disturbance. If such a trajectory does not exist, $\VV(x)$ is positive and $x$ is inside $\frt{\interior{C}}$, described by the following lemma.

\begin{lemma}
\label{lemma:frt_value}
Suppose $C\in\R^n$ is a closed set, and a bounded function $h_C$ satisfies Condition~\ref{assumption:hS}-(1). $\VV(x)$ is positive if and only if $x$ belongs to the FRT of the interior of $C$:
\vspace{-0.5em}
\begin{align}
    \frt{\interior{C}} & = \{x~|~\VV(x)>0\}, \nonumber \\
\{\frt{\interior{C}}\}^c & = \{x~|~\VV(x)=0\}.
    \label{eq:FRT_FRTValueFunction}
\end{align}    
\end{lemma}

\vspace{-1.5em}
\begin{proof} See Appendix \ref{appendix:lemma_frt_value}. 
\end{proof}
\vspace{-0.75em}

\dLee{
As in Lemma~\ref{lemma:frt_value}, $\frt{\interior{C}}$ is characterized as the 
strict zero-superlevel set of $\VV$ in~\eqref{eq:Def_FRT_valueFunction}. If its boundary is differentiable (Condition~\ref{assumption:S}), then the FRT becomes a control invariant set according to Theorem~\ref{thm:frt}. In such cases, whereas the BRT value function \eqref{eq:brt_v_discounted} flattens to zero within the control invariant set, the FRT value function remains positive inside the invariant set.
}

The fundamental nature of reachability, what is called the ``game-of-kind'' \cite{cardaliaguet1996differential, bansal2017hamilton}---whether or not a state is inside the FRT---remains consistent, since \eqref{eq:FRT_FRTValueFunction} holds for any value of $\discount$ and the property is determined by the positivity of $\VV(x)$. However, the discount factor $\discount$ introduces the ``game-of-degree'' aspect \cite{isaacs1999differential} to the reachability problem. In this game-of-degree, $\discount$ serves as a knob that adjusts how conservative the resulting optimal policy will be. As we will see soon, this knob is shaped by the barrier constraint of the CBFs. This is the unique feature of introducing discount in FRT value functions whereas the optimal control policy of a discounted BRT value function is not affected by $\gamma$ within the invariant set due to its flatness \cite{akametalu2018minimum, xue2018reach}.


The introduction of a discount in \eqref{eq:frt_cost} is similar to introducing a discount to infinite-horizon sum-over-time cost optimal control problems \cite{bardi1997optimal}. In fact, many favorable properties of the value function resulting from the discount, including its Lipschitz continuity and the contraction of the corresponding Bellman backup operator, hold similarly.

\begin{proposition}[Lipschitz Continuity]
\label{prop:lipschitz} Suppose $f$ satisfies Assumption \ref{assumption:dynamics} and $h_C$ is Lipschitz continuous. $V_\discount$ is Lipschitz continuous in $\R^n$ if $L_f < \gamma$, where $L_f$ is the Lipschitz constant of $f$. 
\end{proposition}
\vspace{-1em}
\begin{proof} See Appendix \ref{appendix:prop_lipschitz}.
\end{proof}
\vspace{-0.75em}

The condition $L_f < \gamma$ implies that $\gamma$ has to be large enough to suppress the effect of the vector field in prohibiting continuity. Under this condition, since the value function is Lipschitz continuous, it is differentiable almost everywhere by Rademacher's Theorem. 

The contraction property of the Bellman backup will be discussed next after introducing the HJ-PDE characterization of $\VV$, which provides a computational machinery for the computation of the value function $\VV$, and also the uniqueness of the PDE solution.

\vspace{-0.25em}
\subsection{Hamilton-Jacobi Characterization}
\label{subsec:frt-hj}
\vspace{-0.25em}

The HJ-PDE underlying the FRT value function $\VV$ is derived by applying Bellman’s principle of optimality to \eqref{eq:Def_FRT_valueFunction}, which results in the following variational inequality.


\begin{theorem}[Forward Reachable Tube Hamilton-Jacobi
Variational Inequality]
\label{thm:HJPDE}
    Suppose $h_C$ is a bounded and Lipschitz continuous function, and $\discount>0$. $\VV$ in \eqref{eq:Def_FRT_valueFunction} is a unique viscosity solution in $\R^n$ of the following HJ-PDE, called Hamilton-Jacobi forward reachable tube variational inequality (HJ-FRT-VI):
\vspace{-0.75em}
    \begin{align}
        0=\min\Big\{ & \
        \VV(x)- h_C(x), \label{eq:thm_HJPDE} \\
        & \max_u \min_d \frac{\partial \VV}{\partial x}\cdot f(x,u,d) + \discount \VV(x) \Big\}. \nonumber
\vspace{-0.5em}
    \end{align}
\end{theorem}
\vspace{-1.25em}
\begin{proof} See Appendix \ref{appendix:thm_HJPDE}.
\end{proof}
\vspace{-0.75em}

A strictly positive value of $\discount$ guarantees the boundedness and the uniqueness of the solution of the HJ-FRT-VI. (An example in Appendix illustrates these outcomes.) In fact, the uniqueness property follows from the contraction property of the Bellman backup associated with the dynamic programming principle of $\VV$.

To see this, we define a Bellman backup operator $B_T: \text{BUC}(\R^n)\rightarrow \text{BUC}(\R^n)$ for $T>0$, where $\text{BUC}(\R^n)$ represents a set of bounded and uniformly continuous functions: $\R^n\rightarrow \R$, as
\vspace{-1em}
\begin{align}
    B_T[V](x)\!\coloneqq\!\! \sup_{\BWstrategy\in\BWStrategy}\inf_{\BWctrl\in \BWcfset}\!\!\max\Big\{ & \! \max_{t\in[-T,0]} e^{\discount t}h_C (\BWtraj(t)),  \nonumber \\
    & \!\!\!e^{-\discount T} V(\BWtraj(-T)) \Big\}. \vspace{-1.25em}
    \label{eq:frt-bellman-backup}
\end{align}
Then, the following holds.

\begin{theorem}[Contraction mapping \& Uniqueness]\label{thm:contractionMapping}
    \!For $V^{1}, V^{2}\!\in\!\text{BUC}(\R^n)$, 
\vspace{-0.5em}
    \begin{align}
        \lVert B_T[V^1]-B_T[V^2] \rVert_\infty \leq e^{-\discount T} \lVert V^1-V^2 \rVert_\infty,
\vspace{-0.75em}
    \end{align}
    and the FRT value function $\VV$ in \eqref{eq:Def_FRT_valueFunction} is the \textit{unique} fixed-point solution to $\VV\!=\!B_T[\VV]$ for each $T\!>\!0$. 
    Also, for any $V\!\in\!\textnormal{BUC}(\R^n)$, $\lim_{T\rightarrow \infty} B_T[V] = \VV$.
\end{theorem}

\vspace{-1em}
\begin{proof} See Appendix \ref{appendix:contractionMapping}.
\end{proof}
\vspace{-0.75em}

Theorem~\ref{thm:contractionMapping} also provides various ways to compute the FRT value function $\VV$ using the operation $B_T[\cdot]$, which do not require any assumptions for the initial guess of the value function, besides the boundedness and uniform continuity in $\R^n$. For further details of how Theorem~\ref{thm:contractionMapping} can be utilized for finite-horizon-based computation or value iteration, see Appendix~\ref{appendix:contraction}.

\subsection{FRT Value Functions and CBFs}
\label{sec:HJ_FRT_CBF}

We now revisit the forward reachability for control invariant sets, extending the analysis in Section \ref{sec:FRT_CtrlInv}, and draw a connection between the FRT value function $\VV$ and robust CBFs. 




We first derive an optimal policy of $\VV$ from the HJ-FRT-VI \eqref{eq:thm_HJPDE}.

\begin{proposition}[Optimal policy of $\VV$]
\label{thm:opt-policy-frt} Under the assumptions in Theorem~\ref{thm:HJPDE}, we define the set-valued map policy $K_\discount: \R^n \rightarrow 2^U$ as 
\vspace{-0.5em}
\begin{equation}
\label{eq:frt-opt-ctrl}
    \resizebox{0.87\hsize}{!}{$\displaystyle K_\discount(x)\!:=\!\left\{ u\in U: \min_{d\in D} \frac{\partial\VV}{\partial x}\cdot f(x, u, d) + \discount \VV(x) \ge 0 \right\}$}, \vspace{-0.5em}
\end{equation}
where $\VV$ is defined as \eqref{eq:Def_FRT_valueFunction}. Then, $K_\discount(x)$ is non-empty for every $x \in \frt{\interior{C}}$ where $\frac{\partial\VV}{\partial x}$ exists. In addition, if $\VV$ is differentiable, any element of $K_\discount(x)$ is an optimal control input with respect to $\VV$ in \eqref{eq:Def_FRT_valueFunction}. 
\end{proposition}

\vspace{-1em}
\begin{proof} See Appendix \ref{appendix:opt-policy-frt}.
\end{proof}
\vspace{-0.5em}

Note that the non-emptiness of $K_\discount(x)$ is derived from $\VV$ satisfying the HJ-FRT-VI \eqref{eq:thm_HJPDE}. By noting that the second term of the minimum in \eqref{eq:thm_HJPDE} has to be non-negative for \eqref{eq:thm_HJPDE} to hold, we get that 
\vspace{-0.5em}
\begin{equation}\label{eq:FRT_CBFineq}
    \max_{u\in U}\min_{d\in D} \frac{\partial\VV}{\partial x}\cdot f(x, u, d) + \discount \VV(x) \ge 0 \vspace{-0.25em}
\end{equation}
at every $x\in\R^n$ where $\VV$ is differentiable. This corresponds to the \cbfconstraint~in \eqref{eq:rCBF_ctrl_condition_original} where $\alpha(y)=\discount y$.

Since the value function is Lipschitz continuous and differentiable almost everywhere by Proposition \ref{prop:lipschitz}, $\VV$ satisfies \eqref{eq:FRT_CBFineq} almost everywhere in $\frt{\interior{C}}\subset\R^n$. Note that $\VV$ is 0 everywhere outside $\frt{\interior{C}}$. 
If  $\VV$ is differentiable in the FRT, \eqref{eq:FRT_CBFineq} is satisfied everywhere in FRT, which constitutes the robust CBF in Definition \ref{def:rcbf}:

\begin{corollary}
\label{cor:frt_is_cbf}
If $\VV$ is continuously differentiable in $\frt{\interior{C}}$, $\VV$ is a robust CBF, and $\frt{\interior{C}}$ becomes a robust control invariant set.
\end{corollary}


\noindent Next, any valid robust CBF $\cbf$ itself is the FRT value function in $\interior{S}$:

\begin{theorem}[Inverse optimality of CBFs]
\label{thm:cbf_is_frt} 
Let $\cbf:\R^n \rightarrow \R$ be a robust CBF for a closed set $S$, satisfying 
\vspace{-0.5em}
    \begin{equation}\label{eq:CBF_ineqality}
        \max_{u \in U} \min_{d \in D} \dcbf (x) \cdot f(x, u,d) + \discount \cbf(x) \ge 0, \vspace{-0.5em}
    \end{equation}
    for all $x\in S$ and some $\discount > 0$.
    Then,  \vspace{-0.5em}
    \begin{equation}
        \label{eq:cor_cbf_is_frt_eq2} \VV(x) = \max\{0, \cbf(x)\}  \vspace{-0.5em}
    \end{equation}
    is the unique viscosity solution of the HJ-FRT-VI \eqref{eq:thm_HJPDE} with $h_C(x) = \cbf(x)$.
\end{theorem}
\vspace{-1.5em}
\begin{proof} See Appendix \ref{appendix:cbf_is_frt}.
\end{proof}

Corollary~\ref{cor:frt_is_cbf} and Theorem~\ref{thm:cbf_is_frt} establish a tight theoretical linkage between HJ reachability analysis and CBFs, wherein the role of discount factor is crucial. By introducing the discount factor to the reachability formulation, the value function becomes a CBF-like function in that it satisfies the \cbfconstraint~almost everywhere in the set $S$, and in the best case when it is differentiable, it becomes the CBF. Moreover, by Proposition \ref{thm:opt-policy-frt}, the \cbfconstraint~defines the optimal policy. On the other hand, by Theorem~\ref{thm:cbf_is_frt}, any CBF can be interpreted as an FRT value function with a discount factor. 
Thus, the inverse optimal control principle \cite{ABAZAR2020119} underlying the \cbfconstraint~and the CBF is the discounted forward reachability. The discounted FRT cost function \eqref{eq:frt_cost} is the inverse optimal cost that characterizes the CBF itself as the value function, and any control satisfying the \cbfconstraint~as the corresponding optimal control. In \cite{krstic2021inverse}, the inverse optimality of CBF-based min-norm controllers has been investigated as an infinite-horizon running cost problem; however, this work does not capture the inverse optimality of the CBF itself.

Finally, let $W$, a continuously differentiable function, satisfy
\begin{enumerate}
    \item\; \vspace{-0.5em}
    \begin{equation}
    W(x) \ge h_C(x), \label{eq:Wh}        
    \end{equation}
    \item \;\vspace{-1em}
    \begin{equation}
    \max_{u\in U} \min_{d \in D} \nabla W(x) \cdot f(x, u, d) + \gamma W(x) \ge 0,      \label{eq:barrier-W}
    \end{equation}
\end{enumerate}
or, equivalently, it satisfies the variational inequality
\begin{equation}
\resizebox{0.99\hsize}{!}{$\displaystyle\min\Bigl\{W(x)-h_C(x),\ \max_{u\in U} \min_{d \in D} \nabla W(x) \cdot f(x, u, d) + \gamma W(x) \Bigr\} \ge 0$},
\label{eq:frt-pde-w}    
\end{equation}
for all states $x$ at which $W(x) \ge 0$. Such $W$ is a \textit{supersolution} of the HJ-FRT-VI \eqref{eq:thm_HJPDE}.

Such $W$ is a valid control barrier function whose zero-superlevel set, $C_W :=\{x\;|\; W(x) > 0\}$, is a superset of $C$. Then, we can prove that everywhere $W(x)$ is positive, $W(x)$ is lower-bounded by $V(x)$, thus resulting in $C_W$ to be a superset of $\frt{\interior{C}}$:



\begin{theorem}$\VV(x) \le W(x)$ for all $x$ such that $\VV(x) > 0$ or $W(x) > 0$ and $\frt{\interior{C}} \subset C_W$.
\label{thm:frt-cbf-main}
\end{theorem}
\begin{proof} By Theorem~\ref{thm:cbf_is_frt}, for all $x$ such that $\VV(x) > 0$ or $W(x) > 0$, we get
\begin{align*}
    \VV(x) & = \sup_{\BWstrategy\in\BWStrategy}\inf_{\BWctrl\in\BWcfset} \sup_{t\in(-\infty, 0]} e^{\discount t} h_C(\BWtraj(t)) \\
    & \le \sup_{\BWstrategy\in\BWStrategy}\inf_{\BWctrl\in\BWcfset} \sup_{t\in(-\infty, 0]} e^{\discount t} W(\BWtraj(t)) \quad (\text{by}\; \text{\eqref{eq:Wh}}) \\
    & = \max\big\{0, W(x) \big\} \quad (\text{by}\; \text{Theorem~\ref{thm:cbf_is_frt}}) \\
    & = W(x) \quad (\text{when}\; \VV(x) > 0 \; \text{or} \; W(x) > 0) \vspace{-1.5em}
    \end{align*}    
\vspace{-1em}
\end{proof}
\vspace{-0.5em}


Theorems~\ref{thm:cbf_is_frt} and \ref{thm:frt-cbf-main} establish a tight connection between CBFs and FRTs:
\begin{remark}
\label{rmk:supersoln_cbf}
    \begin{enumerate}
        \item The robust control invariant set characterized by any CBF, $C_W$, is an FRT. 
        In particular, this invariant set is the FRT of itself, i.e., $\mathrm{FRT}(C_W) = C_W$ (by Corollary~\ref{cor:frt_corollary}), 
        and there may exist strict subsets $C$ whose FRT coincides with $C_W$.

        \item For any CBF $W$ and any subset $C \subseteq C_W$, $W$ is a supersolution to the HJ-FRT-VI associated with $h_C$ satisfying $h_C(x) \le W(x)$.
    \end{enumerate}
\end{remark}

Theorem~\ref{thm:frt-cbf-main} and Remark~\ref{rmk:supersoln_cbf} imply that any continuously differentiable supersolution of the HJ-FRT-VI is equivalent to a CBF. 
Therefore, general computational frameworks for supersolutions can serve as frameworks for generating CBFs, which we explore next. 
Through our formulation, we can design a robust control invariant set that outer-approximates the FRT, even when the initial set $C$ is not robustly control invariant.

\section{\dLee{Learning-based approach for FRT-CBF}}
\label{sec:frt-cbf-learning}
We now propose a method that uses a neural network to learn a CBF whose control invariant set lies between the given constraint set $X$ and the initial set $C$. This framework is constructed by the mechanisms of Theorems~\ref{cor:overapproximation} and \ref{thm:frt-cbf-main} which verifies the learned control invariant set and the CBF as an outer- and over-approximation of the FRT and the FRT value function, respectively. We consider the CBF learned through this framework the \textit{neural forward reachable tube control barrier function} (FRT-CBF) of the initial set $C$.

\begin{assumption} The given initial set $C$ is a subset of the viability kernel of the constraint set $X$.
\end{assumption}

The neural network we learn outputs two scalar functions, $u_\theta(x)$ and $v_\theta(x)$. Our objective is to learn the FRT-CBF, represented as
\begin{equation}
    W_\theta(x) := h_C(x) + u_\theta^2(x).
    \label{eq:w_represent}
\end{equation}
By taking this representation, we can guarantee $W_\theta (x) \ge h_C (x)$ as a hard constraint.

The learning consists of three objectives:
\begin{enumerate}
    \item Learn $W_\theta(x)$ that satisfies \eqref{eq:frt-pde-w}, so that according to Theorem~\ref{thm:frt-cbf-main}, $C_W$ includes the minimal FRT.
    \item Ensure that
    \[
    W_\theta(x) \le h_{X}(x)
    \]
    is satisfied, so that the learned control invariant set $C_W$ is a subset of the constraint set $X$.
    \item Ensure the control invariance condition, i.e., 
    \begin{equation}
    \max_{u \in U} \min_{d \in D} \frac{\partial W_\theta}{\partial x}(x) \cdot f(x, u, d) \ge 0,
    \vspace{-0.5em}
    \end{equation}
    for all $x \in \partial C_W$, to satisfy the condition in Lemma~\ref{lemma:tangential}.
\end{enumerate}

To ensure the second condition, we use the second output of the neural network, $v_\theta(x)$, and impose a condition that
\begin{equation}
    W_\theta(x) = h_C(x) + u_\theta^2(x) = h_X (x) - v_\theta^2(x) \le h_X (x).
    \label{eq:uv_condition_original}
\end{equation}

In summary, because $u_\theta(x)^2 \ge 0$ and $v_\theta(x)^2 \ge 0$, this parameterization enforces the relation
\begin{equation}
h_C(x) \le W_\theta(x) \le h_X (x),    
\label{eq:set_inclusion_cond}
\end{equation}
and thus
\[
C \subset C_W \subset X.
\]

\subsection{Related Work}
Several recent works are closely related to the proposed training formulation. A number of studies apply the BRT-based CBVF formulation (in Section~\ref{sec:brt}) to learn CBFs using neural networks \cite{manda2024domain, wang2024simultaneous, kim2025reachability}. These approaches are largely inspired by physics-informed machine learning for approximating solutions of PDEs, particularly HJ-PDEs arising in reachability analysis, as in \cite{bansal2021deepreach}. Another line of work draws inspiration from backward reachability and learns CBFs that enlarge the control invariant set through Bellman backup-type updates \cite{so2024train}. 

These bodies of work may suffer from issues such as non-uniqueness, unboundedness, or discontinuity if the discount factor is not treated carefully. Although maximal invariant sets are an appealing property of backward reachability-based frameworks, our theory suggests that such formulations exhibit a fundamental discrepancy with well-behaved CBFs (e.g., Condition~\ref{assumption:hS}). This discrepancy can be addressed by instead adopting a forward reachability perspective.

Some works additionally incorporate pre-verified safe regions, similar to the initial set used in our formulation \cite{dawson2023safe, wang2018permissive}. These predefined safety regions typically serve as an initial guess for the CBF or as a constraint that prevents the learned invariant sets from becoming overly conservative. However, these works do not explicitly interpret the underlying problem through the lens of forward reachability. 

Our work is related to these approaches in that we also use a neural network to approximate a supersolution of the HJ-FRT-VI, and the resulting loss function designs share similarities with those used in several neural CBF learning methods. However, crucially, we explicitly formulate the objective as learning a CBF whose zero-superlevel set outer-approximates the forward reachable tube (FRT) of a given initial set. 





\subsection{Loss Function}

At each state sample $x_i$, we consider the following three loss terms for the learning.

1. \textit{HJ-FRT-VI Supersolution Loss:} We evaluate the violation of the condition \eqref{eq:frt-pde-w} with
\begin{align*}
    \resizebox{1.00\hsize}{!}{$\displaystyle\ell_{\mathrm{frt}}(x_i )
    = \min\Bigl\{
        \max_{u \in U} \min_{d \in D} \frac{\partial W_\theta}{\partial x}(x_i) \cdot f(x_i, u, d) + \gamma W_\theta (x_i),
        \; 0
    \Bigr\}^2.$}
\end{align*}
This penalizes the violation of the inequality \eqref{eq:barrier-W} that $W_\theta$ must satisfy. Note that \eqref{eq:Wh} is already guaranteed by the representation \eqref{eq:w_represent}. This term is active only when $W(x_i) \ge 0$.

2. \textit{Set Inclusion Condition Loss:} The condition \eqref{eq:uv_condition_original} becomes equivalent to learning $u_\theta$ and $v_\theta$ such that
\begin{equation}
    u_\theta^2(x) + v_\theta^2(x) = h_X (x) - h_C(x).
    \label{eq:uv_condition}
\end{equation}
We evaluate the violation of the constraint \eqref{eq:uv_condition} with
\begin{align*}
    \resizebox{1.00\hsize}{!}{$\displaystyle\ell_{\mathrm{inclusion}} (x_i)
    = \left(
            (h_X (x_i) - h_C(x_i)) - \bigl(u_\theta (x_i)^2 + v_\theta (x_i)^2\bigr)
        \right)^2$}
\end{align*}
so that the condition \eqref{eq:set_inclusion_cond} is satisfied.

\begin{figure*}
\centering
\includegraphics[width=\textwidth]{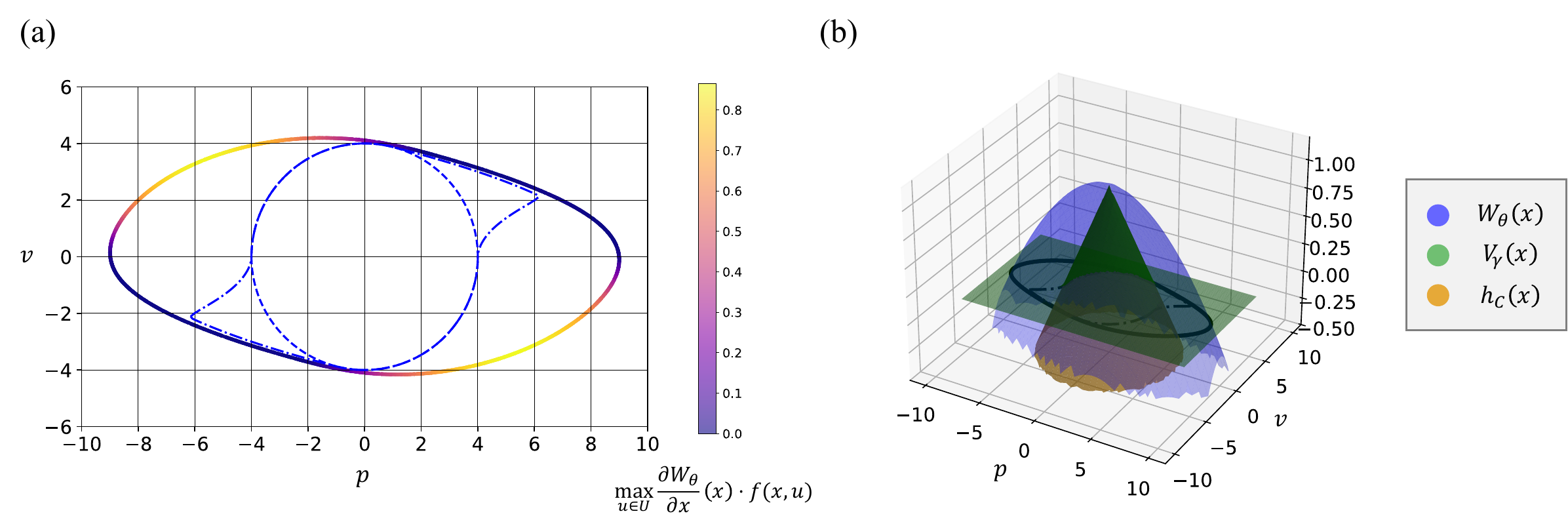}
\vspace{-1em}
\caption{\textbf{(a)} Boundary of the learned FRT-CBF, which includes the initial set $C$ (dashed), and the $\frt{\interior{C}}$ (dash-dotted). The color of the boundary indicates the value of $\max_{u \in U} \min_{d \in D} \frac{\partial W_\theta}{\partial x}(x) \cdot f(x, u, d)$ which must be nonnegative for control invariance. The minimum value occurred is -2.07e-4. \textbf{(b)} The learned $W_\theta(x)$ which is continuously differentiable, in comparison to $\VV(x)$ and $h_C$, illustrating the relationship in Theorem~\ref{thm:frt-cbf-main}.}
\label{fig:learned_di_result}
\vspace{-.5em}
\end{figure*}

3. \textit{Control Invariance Loss:}
We sample states near the boundary of $C_W$, finding a sample that satisfies
\begin{equation}
    0 \le W_\theta (x_i) \le \epsilon, \label{eq:boundary_range}
\end{equation}
where $\epsilon$ is set as a small threshold. We employ rejection sampling to sample the states that satisfy this condition. For these states, we evaluate the following loss
\begin{align*}
    \ell_{\mathrm{ci}} (x_i)
= \min\Bigl\{
        \max_{u \in U} \min_{d \in D} \frac{\partial W_\theta}{\partial x}(x_i) \cdot f(x_i, u, d),
        \; 0
    \Bigr\}^2.
\end{align*}
This penalizes the state that does not satisfy the tangential condition in Lemma~\ref{lemma:tangential}.

The final learning objective is
\begin{align}
    \mathcal{L}
    = & \; w_1 \mathbb{E}_{x \in \{x | W_\theta(x) \ge 0 \}} [\ell_{\mathrm{frt}}(x_i )] + w_2 \mathbb{E}_{x} [\ell_{\mathrm{inclusion}}(x_i )]  \nonumber \\
    & + w_3 \mathbb{E}_{x \in \{x | 0 \le W_\theta(x) \le \epsilon \}} [\ell_{\mathrm{ci}}(x_i )]. 
\end{align}

\subsection{Neural Network \& Training Procedure} 

The model is a multilayer perceptron with input dimension $n$, which takes in normalized state inputs, and output dimension $2$, outputting $(u_\theta(x), v_\theta(x))$. We use the softplus activation,
\begin{align}
    \sigma(z) = \frac{1}{\beta}\log\!\bigl(1 + e^{\beta z}\bigr),
\end{align}
with $\beta > 0$. This results in $u_\theta, v_\theta$ being continuously differentiable functions with differentiable level sets, and the maximum curvature can be controlled by choosing small enough $\beta$ \cite{srinivas2022efficient}.

By the chain rule, we compute the gradient of $W_\theta$ as
\begin{align}
    \frac{\partial W_\theta(x)}{\partial x}
    &=     \frac{\partial h_C(x)}{\partial x} + 2 u_\theta(x)     \frac{\partial u_\theta(x)}{\partial x},
\end{align}
where $\frac{\partial u_\theta(x)}{\partial x}$ is evaluated based on auto-differentiation.

\begin{remark}
    Since we represent the CBF as \eqref{eq:w_represent}, for $W_\theta$ to satisfy Condition~\ref{assumption:hS}, we need to use $h_C$ that also satisfies Condition~\ref{assumption:hS}, meaning that our framework is currently limited to initial sets with a differentiable boundary. Similarly, the constraint target function $h_X$ also has to be a continuously differentiable function in order to satisfy \eqref{eq:uv_condition_original}. In practice, one can smoothen the given non-smooth initial and constraint set target functions. 
\end{remark}

\textit{Training details. }We use a multilayer perceptron with three hidden layers with 512 features, and $\beta = 10$. Each batch for one training iteration is generated by drawing 40,000 i.i.d. samples uniformly from the state range $[-10,10]\times[-10,10]$, and 5,000 samples satisfying \eqref{eq:boundary_range} to evaluate the control invariance loss. We use the loss weights $w_1\!=\!w_2\!=\!w_3\!=\!1$. We used the Adam optimizer with learning rate $\eta = 10^{-5}$, and run training for 10,000 iterations. The training takes around 30 minutes on an NVIDIA RTX 5090 GPU.

\subsection{Result}

We revisit the double integrator example in Section 3.3, and focus on Case (a) where the chosen initial set $C$ is \textit{not} control invariant.

The target function of the initial set $h_C$ is a circular quadratic function centered at the origin with radius $r=4$:
\[
h_C(x) = 1 - \frac{\norm{x}^2}{r^2}
\]
The constraint function $h_X$ is set as a circular quadratic function centered at the origin with radius $R=9$:
\begin{equation*}
    h_X(x) = 1 - \frac{\norm{x}^2}{R^2}.
\end{equation*}
The training loss at the last iteration of the training is 2.25e-7, implying that the CBF is learned successfully with negligible errors. The learned control invariant set and the FRT-CBF is visualized in Figure \ref{fig:learned_di_result}. Together, we also visualize the FRT value function $\VV(x)$, computed by the finite horizon-based method described in Appendix~\ref{appendix:contraction}. We used $\gamma=2$ for both computations.

\vspace{-0.5em}
\section{Discussion}
\label{sec:discussion}
\vspace{-0.5em}




Although inspired by existing works that establish the connection between reachability and robust control invariance, our work is the first paper that connects \textit{forward} reachability to the analysis of robust control invariance. Existing works have focused on backward reachability and proposed various value functions that characterize the inevitable BRT of the failure region, for both finite time \cite{choi2021robust, lygeros2004reachability, Mitchell2005, fisac2015reach} and infinite time \cite{Fisac2018, xue2018reach,akametalu2018minimum, tonkens2022refining}. However, none of these functions are CBFs, as discussed in Section \ref{sec:brt} and summarized in Table \ref{table:comparision_CBF_Reachability}. The main issue with these formulations is that the value function is either flat inside the invariant set (when $\lambda > 0$ in \eqref{eq:brt_v_discounted}), or can become discontinuous (when $\lambda \le 0$), both of which make its use as a CBF infeasible.

Our formulation of forward reachability ensures compliance with the \cbfconstraint~in the CBF definition. Additionally, the value function is both continuous and bounded in $\R^n$, while the corresponding HJ-PDE has a unique solution. The central idea behind our approach is the usage of the discount factor backward in time in the definition of the discounted FRT value function \eqref{eq:Def_FRT_valueFunction}. In contrast to the discount in BRT formulations leading to the emergence of $-\gamma V_\gamma(x)$ in the corresponding HJ-PDEs (Table \ref{table:comparision_CBF_Reachability}, first row), the usage of discount in this way leads to the emergence of a positive $\gamma V_\gamma(x)$ term in the HJ-FRT-VI \eqref{eq:thm_HJPDE}, and thus the satisfaction of the \cbfconstraint~(Table \ref{table:comparision_CBF_Reachability}, last row). Moreover, $e^{\gamma t}$ vanishes as $t$ approaches $-\infty$, thereby ensuring continuous, bounded value functions and the solution uniqueness of the HJ-FRT-VI, resulting from the contraction mapping property outlined in Section \ref{subsec:frt-hj}. 


\vspace{-0.5em}
\section{Conclusions}
\label{sec:conclusion}
\vspace{-0.5em}

In this study, we have presented a framework that establishes a strong linkage between reachability, control invariance, and Control Barrier Functions (CBFs) through a Hamilton-Jacobi differential game formulation. Two main aspects of our approach are the use of the forward reachability concept in lieu of backward reachability, and the incorporation of a discount factor in the value function. These elements induce a contraction in the Bellman backup of the value function, thereby shaping it to satisfy the barrier constraint of the CBFs. Importantly, we note that prior formulations relying on backward reachability were unable to establish this connection between reachability, control invariance, and CBFs. Thus, our work fills a crucial gap in the existing literature, presenting a new perspective on the interplay among these key concepts in the safe control literature through the lens of forward reachability. The emergence of the barrier constraint in the forward reachability formulation opens new avenues for constructing CBFs grounded in reachability analysis.

One salient condition underlying our study is the differentiability of the boundaries of control invariant sets. Although this condition is required for the zero-superlevel sets of any valid CBFs, recent extensions of CBFs can alleviate this condition \cite{cortes2024mandalay, hirsch2025viscosity}. A deeper understanding of the implications of this condition is crucial for broadening the applicability of our results. For example, the relevance of such regularity properties of the boundary of the safe set for stochastic systems is elucidated in \cite{buckdahn2019viability}. \finalrev{Finally, applying the proposed neural FRT-CBF framework to more practical high-dimensional systems and formally verifying the resulting framework are left as important directions for future research.}

\section{Appendix}

\subsection{Proof of Theorem \ref{thm:frt}}
\label{appendix:thm_frt}

For the proof of Theorem~\ref{thm:frt}, we have to reason about the backward flow of the dynamics, as in \eqref{eq:backward_dynamics}. Thus, we first consider the notion of backward control invariance, the mirrored version of the forward control invariance.

\begin{definition}[Robustly \textit{backward} control invariant] A set $S\subset \R^n$ is robustly \textit{backward} control invariant (under \eqref{eq:dynamics}) if for all $x \in S$, for all $\BWstrategy\in\BWStrategy$, for any $\epsilon > 0$ and time $T > 0$, there exists a backward control signal $\BWctrl\in\BWcfset$ such that $\BWtraj(t)\in S + B_\epsilon$ for all $t\in[-T,0]$, where $\BWtraj$ solves \eqref{eq:backward_dynamics}.
\end{definition}

\vspace{-0.5em}

Put simply, backward control invariant sets are forward control invariant under the negated dynamics (where the time flows inversely). We can characterize the backward control invariant sets similarly to Lemma~\ref{lemma:tangential}:


\begin{corollary}
\label{corollary:tangential_backward}
A set $S$ satisfying Condition~\ref{assumption:S} is robustly backward control invariant if and only if for all $x\in \partial S$,
\vspace{-1em}
\begin{equation}
\label{eq:nrci_tangent}
    \exists u \in U ~\textnormal{s.t.}~ -\dcbf(x) \cdot f(x, u, d) \ge 0 \; \forall d\in D,
\vspace{-1em}
\end{equation}
for $h_S$ satisfying Condition~\ref{assumption:hS}.
\end{corollary}

\noindent By combining Lemma~\ref{lemma:tangential} and Corollary~\ref{corollary:tangential_backward}, we draw a connection between forward and backward invariant sets.

\begin{lemma}\label{lemma:positive_negative} 
A set $S$ satisfying Condition~\ref{assumption:S} is robustly forward control invariant if and only if $\interior{S}^c$ is robustly backward control invariant.
\end{lemma}

\vspace{-1em}

\begin{proof} 
By Lemma~\ref{lemma:tangential}, $S$ is robustly forward control invariant if and only if for all $x\in \partial S$, \eqref{eq:rci_tangent} is satisfied. Note that $\partial S = \partial \interior{S}^c$ and $\interior{S}^c$ and $h_{\interior{S}^c}:=-\cbf$ also satisfy Conditions~\ref{assumption:S} and \ref{assumption:hS}, respectively. Since $\frac{\partial h_{\interior{S}^c}}{\partial x}(x) = - \frac{\partial \cbf}{\partial x}(x)$, \eqref{eq:rci_tangent} is equivalent to
\begin{equation}
\label{eq:rci_tangent_negative}
    \resizebox{0.87\hsize}{!}{$\displaystyle\exists u \in U ~\textnormal{s.t.}~ \frac{\partial h_{\interior{S}^c}}{\partial x}(x) \cdot \left(-f(x, u, d)\right) \ge 0 \; \forall d\in \dset.$}
\end{equation}
By applying Corollary~\ref{corollary:tangential_backward}, $\interior{S}^c$ is robustly backward control invariant if and only if $\forall x\!\in\! \partial S$, \eqref{eq:rci_tangent_negative} holds.
\end{proof}

\vspace{-0.5em}

\noindent In Lemma~\ref{lemma:positive_negative}, Condition~\ref{assumption:S} guarantees that, for a state $x_1$ on the boundary of $S$, if there exists a particular control $u_1$ such that $f(x_1,u_1,d)$ points inward to $S$ for all $d\in\dset$, $-f(x_1,u_1,d)$ points outwards to $S$ for all $d\in\dset$.

Next, we introduce the concept of a viability kernel under the backward dynamics:

\begin{definition} \label{def:vk}
A viability kernel of a closed set $C \subset \R^n$ under the backward dynamics \eqref{eq:backward_dynamics}, is defined as:
{\small
\vspace{-1.75em}
\begin{align*}
\label{eq:vk}
\begin{split}
     \vk{C} & \coloneqq \{  x \in \R^n ~|~\forall \BWstrategy\in\BWStrategy,\epsilon>0,T>0,  \exists \BWctrl \in \BWcfset \text{ s.t. } \\
    &\!\!\!\! \forall t\in[-T,0], \BWtraj(t)\in C + B_\epsilon, \textrm{ where } \BWtraj \textrm{ solves } \eqref{eq:backward_dynamics}. \}
\end{split}
\end{align*}}
\vspace{-1.75em}
\end{definition}

This definition applies the concept of \textit{leadership kernel} in \cite{cardaliaguet1996differential} to the backward dynamics. Given any closed set $C\subset\R^n$, the $\vk{C}$ is the largest robustly backward invariant set in $C$. 

The FRT and the viability kernel under the backward dynamics have the following complement property.
\begin{lemma} 
\label{lemma:frt_partition}
For any open set $C \subset \R^n$, the state space $\R^n$ can be partitioned into $\frt{C}$ and $\vk{C^c}$, i.e.,
\vspace{-0.5em}
\begin{equation}
\label{eq:frt_partition}
\{\frt{C}\}^c = \vk{C^c}.
\vspace{-0.5em}
\end{equation}
Furthermore, $\vk{C}$ is always a closed set and $\frt{C}$ is always an open set.
\end{lemma}
\vspace{-1em}
\begin{proof} 
The complement relationship \eqref{eq:frt_partition} follows directly from Definitions \ref{def:frt} and \ref{def:vk}. That $\vk{C}$ is a closed set is proven in \cite{cardaliaguet1996differential}, and $\frt{C}$ being an open set results from the complement relationship.
\end{proof}


Based on the lemmas above, we present the proof of Theorem~\ref{thm:frt} below.
\vspace{-1em}
\begin{proof} \dLee{(Theorem \ref{thm:frt})
Since $\vk{C^c}$ is the \finalrev{maximal} robustly backward invariant set contained in $C^c$, Lemmas~\ref{lemma:positive_negative} and~\ref{lemma:frt_partition} imply that its complement $\frt{C}$ is the \finalrev{minimal} robustly forward control invariant set containing $C$.
}
\end{proof}
\vspace{-0.75em}

\subsection{Proof of Corollary \ref{cor:frt_corollary}}
\label{appendix:proof-frt-fixed-point}
\dLee{
By Lemma~\ref{lemma:positive_negative}, a set $S$ satisfying Condition~\ref{assumption:S} is robustly control invariant if and only if the complement of $\interior{S}$ is robustly backward control invariant. By Definition~\ref{def:vk}, this is equivalent to $\vk{\{\interior{S}\}^c} = \{\interior{S}\}^c$. By Lemma~\ref{lemma:frt_partition}, this condition is further equivalent to $\frt{\interior{S}} = \interior{S}$.
}

\subsection{Proof of Theorem \ref{cor:overapproximation}}
\label{appendix:proof-frt-overapproximation}
We want to prove that $\frt{\interior{C}} \subset S$, which is sufficient to prove $\{\interior{S}\}^c \subset \{\frt{\interior{C}}\}^c$. Since the set $S$ satisfies Condition~\ref{assumption:S} and is robustly control invariant, by Lemma~\ref{lemma:positive_negative}, $\{\interior{S}\}^c$ is robustly backward control invariant. Also, since $C \subset S$, $\{\interior{S}\}^c \subset \{\interior{C}\}^c$. By Lemma~\ref{lemma:frt_partition} $\{\frt{\interior{C}}\}^c = \vk{\{\interior{C}\}^c}$, which is the largest robustly backward invariant in $\{\interior{C}\}^c$. Therefore, $\{\interior{S}\}^c \subset \{\frt{\interior{C}}\}^c$.

\subsection{Proof of Lemma \ref{lemma:frt_value}}
\label{appendix:lemma_frt_value}

We define two one-to-one functions: $\rho_\ctrl:\cfset\rightarrow\BWcfset$:
\vspace{-0.5em}
\begin{equation}
\label{eq:one-to-one-ctrl}
    \rho_\ctrl(\ctrl) (-t) = \ctrl(t),\quad \forall t\in[0,\infty); \vspace{-0.5em}
\end{equation}
and $\rho_{\strategy}:\Strategy\rightarrow\BWStrategy$:
\vspace{-0.5em}
\begin{equation}\label{eq:one-to-one-dstb}
    \rho_{\strategy}(\strategy)[\BWctrl] (-t) = \strategy[\rho_\ctrl^{-1}(\BWctrl)](t), \vspace{-0.5em}
\end{equation}
for all $\BWctrl\in\BWcfset$ and $t\in[0,\infty)$.
Then, \vspace{-0.5em}
\begin{equation}
\label{eq:proof_lemma3_2}
    \BWtraj(-t)=\traj(t) \quad \forall t\in[0,\infty), \vspace{-0.5em}
\end{equation}
where $\BWtraj$ solves \eqref{eq:backward_dynamics}, and $\traj$ solves 
\vspace{-0.25em}
\[\dot{\traj}(t)=-f(\traj(t),\rho_\ctrl^{-1}(\BWctrl)(t), \rho_{\strategy}^{-1}(\BWStrategy)[\rho_\ctrl^{-1}(\BWctrl)](t)) \vspace{-0.25em}
\]
for $t>0$, and $\traj(0)=x$.
Since $\rho_\ctrl$ and $\rho_{\strategy}$ are one-to-one, the viability kernel of $\interior{C}^c$ under $-f$, defined in Definition \ref{def:vk}, is equivalent to the following set: \vspace{-0.5em}
\begin{align}
     \vk{C} = & \{ x \in \R^n ~|~\forall \strategy\in\Strategy,\epsilon>0,T>0, \exists \ctrl \in \cfset \text{ s.t.} \nonumber \\
     &  \forall t\in[0,T], \traj(t)\in C + B_\epsilon \}, \vspace{-0.5em}
\label{eq:proof_lemma3_1}
\end{align}
where $\traj$ solves \vspace{-0.5em}
\begin{align}
    \label{eq:forward_dyn_ctrl_dist}
    \resizebox{.85\hsize}{!}{$\dot{\traj}(t)=-f(\traj(t),\ctrl(t), \strategy[\ctrl](t)), \;\forall t>0,\;\; \traj(0)=x.$} \vspace{-0.5em}
\end{align}
By \cite[Lemma 1]{xue2018reach}, the viability kernel \eqref{eq:proof_lemma3_1} of $\interior{C}^c$ is characterized by a particular value function:
\begin{align*}
     \resizebox{\hsize}{!}{$\displaystyle \vk{\interior{C}^c}\!=\!\Big\{x~\Big|~\! \inf_{\strategy\in\Strategy}\sup_{\ctrl\in\cfset}\inf_{t\in[0,\infty)} e^{-\gamma t}\big(\!-\!h_C (\traj(t))\big)\!=\!0\Big\},$}
\end{align*}
where $\traj$ solves \eqref{eq:forward_dyn_ctrl_dist}.
By \eqref{eq:proof_lemma3_2} and one-to-one properties of $\rho_\ctrl$ and $\rho_{\strategy}$,~$\vk{\interior{C}^c}=\{x~|~ \VV (x) =0 \}$.
By Lemma~\ref{lemma:frt_partition},
\vspace{-1em}
\begin{align}
    \frt{\interior{C}} =\{x~|~ \VV (x) \neq 0 \}.
    \label{eq:proof_lemma3_4}
\vspace{-0.5em}
\end{align}
Since $h_C$ is bounded, 
\begin{align}
    \VV(x) \geq \sup_{\BWstrategy\in \BWStrategy}\inf_{\BWctrl\in\BWcfset} [e^{\gamma t}h_C(\BWtraj(t)) ~|~t=-\infty]=0,
    \label{eq:lemma_nonNegative_V}
\end{align}
for all $x\in\R^n$.
By combining \eqref{eq:proof_lemma3_4} and \eqref{eq:lemma_nonNegative_V}, $\frt{\interior{C}}\!=\!\{ x~|~ \VV(x)>0 \}$.

\dLee{This also implies that $\{\frt{\interior{C}}\}^c = \{ x \mid \VV(x) \le 0 \}$. From the definition of the FRT value function in~\eqref{eq:Def_FRT_valueFunction}, we have $\VV(x) \ge h_C(x)$ for all $x$, since $\JJ(x, \BWctrl, \BWstrategy) \ge h_C(e^{\gamma 0}\BWtraj(0)) = \cbf(x)$. Therefore, it follows that $\{\frt{\interior{C}}\}^c = \{ x \mid \VV(x) = 0 \}$.
}

\subsection{Proof of Proposition \ref{prop:lipschitz}}
\label{appendix:prop_lipschitz}

For $x_1\in\R^n$ and $\epsilon>0$, there exists $\hat{\xi}_d^-\in\BWStrategy$ such that
\vspace{-0.5em}
\begin{equation*}
    \VV(x_1)\leq \inf_{\BWctrl\in\BWcfset} \JJ(x_1 ,\BWctrl,\hat{\xi}_d^-) +\epsilon, \vspace{-0.5em}
\end{equation*}
where $\JJ$ is defined in \eqref{eq:frt_cost}.
Hence, 
\vspace{-0.5em}
\begin{equation*}
    \VV(x_1)\leq \JJ(x_1 ,\BWctrl,\hat{\xi}_d^-) +\epsilon \vspace{-0.5em}
\end{equation*}
for any $\BWctrl\in\BWcfset$. For $x_2$, there exists $\hat{\mathrm{u}}^-\in\BWcfset$ such that \vspace{-0.5em}
\begin{equation*}
    \VV(x_2) \geq \inf_{\BWctrl\in\BWcfset} \JJ(x_2 ,\BWctrl, \hat{\xi}_d^-) \geq \JJ(x_2,\hat{\mathrm{u}}^-, \hat{\xi}_d^- ) -\epsilon. \vspace{-0.5em}
\end{equation*}
By combining the above two inequalities, we have \vspace{-0.25em}
\begin{align*}
    & \VV(x_1)-\VV(x_2) \leq \JJ(x_1 ,\hat{\mathrm{u}}^- ,\hat{\xi}_d^-) - \JJ(x_2 , \hat{\mathrm{u}}^- ,\hat{\xi}_d^-) + 2 \epsilon \\
    & = \sup_{t\in(-\infty,0]}\!\!e^{\discount t} h_C(\mathrm{x}_1^- (t)) \; - \!\!\sup_{t\in(-\infty,0]}\!\! e^{\discount t} h_C (\mathrm{x}_2^- (t)) + 2\epsilon, \vspace{-0.5em} \notag 
\end{align*}
where $\mathrm{x}_1^-$ solves \eqref{eq:backward_dynamics} for $(\hat{\mathrm{u}}^-,\hat{\xi}_d^- )$ with the terminal state $x_1$, and $\mathrm{x}_2^-$ solves \eqref{eq:backward_dynamics} for $(\hat{\mathrm{u}}^-,\hat{\xi}_d^- )$ with the terminal state $x_2$.
Since there exists $\hat t \in (-\infty, 0]$ such that \vspace{-0.5em}
\begin{equation*}
    \sup_{t\in(-\infty,0]}e^{\discount t}h_C (\mathrm{x}_1^- (t)) \leq e^{\discount \hat t}h_C (\mathrm{x}_1^- (\hat t)) +\epsilon, \vspace{-0.5em}
\end{equation*}
this implies \vspace{-0.5em}
\begin{align*}
    & \VV(x_1)-\VV(x_2) \leq e^{\discount \hat t}h_C (\mathrm{x}_1^- (\hat t)) - e^{\discount \hat t}h_C (\mathrm{x}_2^- (\hat t)) +3\epsilon \\
    & \leq L_{h_C}e^{\discount \hat t} e^{-L_f \hat t} \lVert x_1-x_2\rVert + 3\epsilon \leq L_{h_C}\lVert x_1-x_2\rVert + 3\epsilon, \vspace{-0.5em}
\end{align*}
where $L_{h_C}$ is the Lipschitz constant of $h_C$. The second inequality is a result of Gronwall's inequality, and the third is a result of the condition, $L_f < \gamma$. Using a similar argument, we can show $\VV(x_2)-\VV(x_1)\leq L_{h_C}\lVert x_1-x_2\rVert + 3\epsilon$, thus, $|\VV(x_1)-\VV(x_2)|\leq L_{h_C}\lVert x_1-x_2\rVert+3\epsilon$. Since the previous inequality holds for all $\epsilon>0$, $|\VV(x_1)-\VV(x_2)|\leq L_{h_C}\lVert x_1-x_2\rVert$.

\subsection{Proof of Theorem \ref{thm:HJPDE}}
\label{appendix:thm_HJPDE}

Using the one-to-one mappings $\rho_\ctrl$ and $\rho_{\strategy}$ in \eqref{eq:one-to-one-ctrl} and \eqref{eq:one-to-one-dstb}, the value function $\VV$ in \eqref{eq:Def_FRT_valueFunction} can be written as
\begin{align}
    \VV(x) = \sup_{\strategy\in\Strategy}\inf_{\ctrl \in \cfset} \sup_{t\in[0,\infty)} e^{-\gamma t}h_C(\traj(t)),
\label{eq:v_differeent_def}
\end{align}
where $\traj$ solves the negated dynamics, \eqref{eq:forward_dyn_ctrl_dist}. Then, by adopting the results in \cite{xue2018reach}, we can first present Bellman’s principle of optimality of $\VV$:
\begin{theorem}\label{thm:DP}
(Dynamic Programming principle \cite[Lemma 3]{xue2018reach}).
Suppose $\discount>0$. For $x\in\R^n$,
\begin{align}\VV(x)\!=\!\!\sup_{\BWstrategy\in\BWStrategy}\inf_{\BWctrl\in\BWcfset}\!\!\max\!\Big\{ & \max_{t\in[-T,0]} e^{\discount t}h_C\left(\BWtraj(t)\right), \label{eq:eq_thm_DP} \\
& e^{-\discount T} \VV\left(\BWtraj(-T)\right) \!\Big\} \nonumber
\end{align}
for any $T>0$, where $\BWtraj$ solves \eqref{eq:backward_dynamics}.
\end{theorem}

Next, Theorem~\ref{thm:HJPDE} holds by \cite[Lemma 3]{xue2018reach}, since $\VV$ rewritten as \eqref{eq:v_differeent_def} is the unique viscosity solution to 
\begin{align*}
    0=\min\big\{ & \VV(x)- h_C(x), \\
    & -\min_{u\in U}\max_{d\in D} \frac{\partial \VV}{\partial x}\cdot (-f(x,u,d)) + \gamma \VV(x) \big\}
\end{align*}
in $\R^n$, which is equivalent to \eqref{eq:thm_HJPDE}.

\subsection{Proof of Theorem \ref{thm:contractionMapping}}
\label{appendix:contractionMapping}

Define $l(\strategy,\ctrl,x) \coloneqq \max_{t\in[-T,0]} e^{\gamma t}h_C (\traj(t))$, and \vspace{-0.5em}
\begin{equation*}
    l^i(\strategy,\ctrl,x) \coloneqq e^{-\gamma T} V^i(\traj(-T)), \vspace{-0.5em}
\end{equation*}
for $i=1,2$. Then, \vspace{-0.5em}
\begin{equation*}
    B_T[V^i] = \sup_{\strategy\in\Strategy}\inf_{\ctrl \in \cfset}\max\{ l(\strategy,\ctrl),l^i(\strategy,\ctrl) \}. \vspace{-0.5em}
\end{equation*}
Without loss of generality, let $B_T[V^1](x) \geq B_T[V^2](x)$. For any $\epsilon>0$, $\exists \bar\strategy$ such that $B_T[V^1] -\epsilon < \inf_\ctrl$ $\max\{ l(\bar\strategy,\ctrl),l^1(\bar\strategy,\ctrl) \} $, and $\exists \bar\ctrl$ such that $\inf_\ctrl\max\{$ $ l(\bar\strategy,\ctrl),l^2(\bar\strategy,\ctrl) \}\!+\!\epsilon \!>\!\max\{ l(\bar\strategy,\bar\ctrl),l^2(\bar\strategy,\bar\ctrl)\}$.
Then, \vspace{-0.5em}
\begin{align*}
    & B_T[V^1](x) - B_T[V^2](x) \\
    & \; < 2\epsilon + \max\{ l(\bar\strategy,\bar\ctrl),l^1(\bar\strategy,\bar\ctrl) \} - \max\{ l(\bar\strategy,\bar\ctrl),l^2(\bar\strategy,\bar\ctrl) \} \notag\\
    & \;  \leq 2\epsilon + \lvert l^1(\bar\strategy,\bar\ctrl)-l^2(\bar\strategy,\bar\ctrl) \rvert \notag\\
    &  \; \leq 2\epsilon + e^{-\gamma T}\max_{x\in\R^n}\lvert V^1(x)-V^2(x) \rvert \vspace{-1em}
\end{align*}
The second inequality holds since, for all $a,b,c\in\R$, $|\max\{a,b\}-\max\{a,c\}|\leq|b-c|$. Since the above inequality holds for all $x\in\R^n$ and $\epsilon>0$,
\vspace{-0.5em}
\begin{equation*}
    \lVert B_T[V^1] - B_T[V^2] \rVert_{L^\infty(\R^n)} \leq e^{-\gamma T} \lVert V^1 - V^2 \rVert_{L^\infty(\R^n)} \vspace{-0.5em}
\end{equation*}
Since $\VV$ is a fixed-point solution for all $T>0$, the Banach's contraction mapping theorem \cite[Chapter 9.2]{evans2010partial} implies that $\VV$ is the unique fixed-point solution to $B_T[\VV](x) = \VV(x)$ for all $T>0$. In addition, we have \vspace{-0.5em}
\begin{equation*}
    \lVert B_T[V] - \VV \rVert_{L^\infty(\R^n)} \leq e^{-\gamma T} \lVert V - \VV \rVert_{L^\infty(\R^n)} \vspace{-0.5em}
\end{equation*}
for all $V\!\in\!\textnormal{BUC}(\R^n)$, thus, $\lim_{T\rightarrow \infty} B_T[V] = \VV$.

\subsection{Computation methods for $\VV$}
\label{appendix:contraction}
Theorem~\ref{thm:contractionMapping} allows the computation of $\VV$ with a finite-horizon HJ-PDE, which is guaranteed to converge to $\VV$ as $T \rightarrow \infty$.
Also, the theorem enables other numerical schemes based on time discretization, such as value iteration. 

First, the following lemma presents a finite-horizon HJ equation for the computation of $\VV$.
\begin{lemma}[Finite horizon HJ-PDE for the computation of $\VV$]
\label{lemma:numerical_finite_HJPDE}
For a given initial value function candidate $V^0\in\textnormal{BUC}(\R^n)$, let $W:[0,T]\times\R^n\rightarrow\R$ be the unique viscosity solution to the following initial-value HJ-PDE

\;
\vspace{-2em}
    \begin{align}
        &W(0,x)=\max\{h_S(x) , V^0(x)\}, \text{ for } x\in \R^n, \label{eq:numerical_finite_HJPDE_initialCondition}\\
        &0=\min\bigg\{ W(t,x)- h_S(x), \label{eq:numerical_finite_HJPDE} \\
        & \quad \quad \frac{\partial W}{\partial t} +  \max_u \min_d \frac{\partial W}{\partial x}\!\cdot\!f(x,u,d) + \discount W(t,x) \bigg\}  \nonumber
    \end{align}
    for $(t,x)\in(0,T)\times\R^n$. Then, $W(T,x) \equiv B_T[V^0](x) $.
\end{lemma}

\begin{proof}
We will derive the HJ equation for another value function $W^+$ defined below, and then replace $W^+$ by $W$.
Define $W^+:[-T,0]\times\R^n\rightarrow\R$
\vspace{-0.5em}
\begin{align}
    W^+(t,x)=\inf_{\strategy\in\Strategy}\sup_{\ctrl \in \cfset}\min\bigg\{ & \min_{s\in[t,0]}e^{-\gamma(s-t)}(-h_S(\traj(s))), \nonumber \\
    & e^{\gamma t}(-V_0(\traj(0)))\bigg\},
\label{eq:proof_Lemma_numerical_finiteHJPDE_eq1}
\end{align}
where $\traj$ solves \eqref{eq:forward_dyn_ctrl_dist}.
Then, $W(T,x)=-W^+(-T,x)$, and 
\begin{align}
\begin{split}
    &\resizebox{0.87\hsize}{!}{$\displaystyle W(t,x)=-W^+(-t,x), \frac{\partial W}{\partial t}(t,x) = \frac{\partial W^+}{\partial t}(-t,x),$} \\
    &\resizebox{0.87\hsize}{!}{$\displaystyle\frac{\partial W}{\partial x}(t,x) = -\frac{\partial W^+}{\partial x}(-t,x), \forall (t,x)\in(0,T)\times\R^n.$}
\end{split}\label{eq:proof_Lemma_numerical_finiteHJPDE_eq22}
\vspace{-1em}
\end{align}
We adopt the results in \cite{choi2021robust} for the rest of the proof. The value function in \cite{choi2021robust} is \vspace{-0.5em}
\begin{equation*}
\inf_{\strategy\in\strategy}\sup_{\ctrl}\min_{s\in[t,0]}e^{-\gamma(s-t)}(-h_S(\traj(s))), \vspace{-0.5em}
\end{equation*}
which is exactly the same as \eqref{eq:proof_Lemma_numerical_finiteHJPDE_eq1} except the second term of the minimization operation: $e^{\gamma t}(-V_0((\traj(0)))$. This term affects the terminal condition of $W^+$ but not the dynamic programming principle. 
Thus, $W^+$ and the value function in \cite{choi2021robust} solve the same dynamic programming principle, but their terminal conditions are different. Note that we assume $\gamma>0$, but \cite{choi2021robust} assumes $\gamma\leq0$. However, the sign of $\gamma$ does not affect any arguments in \cite{choi2021robust}'s lemmas and theorems.

Using similar arguments as in the proof of \cite[Theorem 2]{choi2021robust}, we get \vspace{-0.5em}
\begin{align*}
    W^+(t,x)\!=\! \inf_{\strategy\in\Strategy}\sup_{\ctrl \in \cfset} \min\bigg\{& \!\min_{s\in[t,t+\delta]}e^{-\gamma(s-t)}(-h_S(\traj(s))), \\
    & e^{-\gamma\delta} W^+(t+\delta, \traj(t+\delta) ) \bigg\}. \vspace{-0.5em}
\end{align*}
Then, \cite[Theorem 3]{choi2021robust} implies that $W^+$ is the unique viscosity solution to the terminal-value HJ equation:\vspace{-1.5em}

{\small
\begin{align*}
    &W^+(0,x)=-\max\{h_S(x),V_0(x)\} \quad \text{on } \{t=0\}\times\R^n,\\
    &0=\min\Big\{ -\!h_S(x)\!-\!W^+(t,x), \\ & \quad \quad \quad \frac{\partial W^+}{\partial t} +\max_{u}\min_d \frac{\partial W^+}{\partial x}\cdot (-f)(x,u,d)-\gamma W^+(t,x) \Big\}, \vspace{-1em}
\end{align*}
}

\vspace{-1em}
\noindent in $(-T,0)\times\R^n$. By applying \eqref{eq:proof_Lemma_numerical_finiteHJPDE_eq22}, we get the conclusion that $W$ is the unique viscosity solution to \eqref{eq:numerical_finite_HJPDE}.
\end{proof}

In Lemma~\ref{lemma:numerical_finite_HJPDE}, any $V^0\in\textnormal{BUC}(\R^n)$ works for the computation of $\VV$; for instance, a straightforward choice of $V^0$ can be $h_S$. As $T\rightarrow\infty$, $\frac{\partial W}{\partial t}$ vanishes to 0 for all $x\in\R^n$.

Combining Theorem~\ref{thm:contractionMapping} and Lemma~\ref{lemma:numerical_finite_HJPDE}, we have \vspace{-0.5em}
\begin{align}
    \lim_{T\rightarrow \infty}B_T[V^0]=\lim_{T\rightarrow \infty} W(T,x)=V_\discount(x).
\end{align}
The PDE \eqref{eq:numerical_finite_HJPDE} can be numerically solved forward in time from the initial condition \eqref{eq:numerical_finite_HJPDE_initialCondition}, by using well-established time-dependent level-set methods \cite{Mitchell2005b}. 

Theorem~\ref{thm:contractionMapping} also enables other numerical schemes that are based on time discretization, like value iteration, to produce an accurate solution of $\VV$. The following corollary of Theorem~\ref{thm:contractionMapping} provides the guarantee that the value iteration with any initial guess of $V^0 \in\textnormal{BUC}(\R^n)$ will converge to $\VV$ with a Q-linear convergence rate specified by \eqref{eq:lemma_contraction_eq3}. For a given time step size $\Delta t$, the semi-Lagrangian approximation can be applied to the exact Bellman backup operator in \eqref{eq:frt-bellman-backup} for its numerical approximation:
\vspace{-5mm}
\small
\dLee{
\begin{align*}
    V^{k+1}(x)=\max\{\cbf(x),\min_{u\in U}\max_{d\in D}e^{-\gamma \Delta t}V^k (x-\Delta t f(x,u,d))\},
\end{align*}
}
\vspace{-3mm}
\normalsize
\\\noindent
and the resulting value function will converge to $\VV$ when $\Delta t\rightarrow 0$ \cite{akametalu2018minimum}.

\begin{corollary}[Value Iteration]
    For any $V^0 \in \text{BUC}(\R^n)$ and a time step $\Delta t>0$, define the sequence $\{V^k\}_{k=0}^{\infty}$ by an iteration $V^k := B_{\Delta t}[V^{k-1}]$ for $k\in \N$. Then,
    \begin{equation}
        \frac{\lVert V^{k+1} - V_\discount \rVert_\infty}{\lVert V^{k} - V_\discount \rVert_\infty}  \le e^{-\discount \Delta t} < 1,
        \label{eq:lemma_contraction_eq3}        
    \end{equation}
    and thus, $\lim_{k\rightarrow \infty} V^k = \VV$.
\end{corollary}

\begin{proof} This is a direct outcome of Theorem~\ref{thm:contractionMapping}.
\end{proof}

\subsection{Proof of Proposition \ref{thm:opt-policy-frt}}
\label{appendix:opt-policy-frt}

(i) At $x\in\frt{\interior{C}}$ where $\VV$ is differentiable, the HJ-FRT-VI \eqref{eq:thm_HJPDE} implies that $K_\gamma$ is non-empty. 

\noindent
(ii) For any control policy $\pi=\pi(x)\in K_\gamma(x)$, where $\VV$ is differentiable, consider the following equation for $V^\pi_\gamma$: \vspace{-0.5em}
\begin{equation}
\label{eq:proof_prop3_eq1}
\resizebox{0.87\hsize}{!}{$\displaystyle 0\!=\!\min\Big\{\!\
        V^\pi_\gamma(x)- h_C(x),  \min_d \frac{\partial V^\pi_\gamma}{\partial x}\cdot f(x,\pi(x),d) + \discount V^\pi_\gamma(x) \Big\}.$} \vspace{-0.5em}
\end{equation}
For each $x\in\R^n$, $\min\{y-h_C(x), \min_d \frac{\partial \VV}{\partial x}\cdot f(x,\pi(x),d) + \discount y\}$ is monotonically increasing in $y$, so the equation \eqref{eq:proof_prop3_eq1} has a unique solution. 
Also, from the HJ-FRT-VI \eqref{eq:thm_HJPDE}, \vspace{-0.5em}
\begin{align}
    & \resizebox{0.9\hsize}{!}{$\displaystyle 0\!=\!\min\!\Big\{\!\
        \VV(x)\!-\!h_C(x),  \max_u \min_d \frac{\partial \VV}{\partial x}\cdot f(x,u,d) + \discount V_\gamma(x) \Big\} \geq$} \notag\\
        & \resizebox{0.87\hsize}{!}{$\displaystyle \min\!\Big\{ \!\
        \VV(x)\!-\!h_C(x),  \min_d \frac{\partial \VV}{\partial x}\cdot f(x,\pi(x),d) + \discount V_\gamma(x) \Big\}\!\geq\!0.$} \vspace{-1em} \label{eq:proof_prop3_eq2}
\end{align}
The last inequality holds since $V_\gamma -h_C \geq 0$ from \eqref{eq:thm_HJPDE} and $\min_d \frac{\partial \VV}{\partial x}\cdot f(x,\pi(x),d) + \discount V_\gamma(x) \geq0$ since $\pi(x)\in K_\gamma(x)$. Equation \eqref{eq:proof_prop3_eq2} and the uniqueness of \eqref{eq:proof_prop3_eq1} imply $\VV\equiv V^\pi_\gamma$ for any $\pi$. By replacing $\VV$ by $V^\pi_\gamma$ in \eqref{eq:proof_prop3_eq1}, \vspace{-0.5em}
\begin{equation*}
\resizebox{\hsize}{!}{$\displaystyle 0\!=\!\min\Big\{\!\
    V^\pi_\gamma(x)\!-\!h_C(x),  \min_d \frac{\partial V^\pi_\gamma}{\partial x}\cdot f(x,\pi(x),d) + \discount V^\pi_\gamma(x) \Big\}$.} \vspace{-0.5em}
\end{equation*}
The solution to the above PDE can be considered as the value function \eqref{eq:frt_cost} under $\pi(x)$ and worst-case disturbance, and since $\VV\equiv V^\pi_\gamma$, we conclude that any control $u\in K_\gamma(x)$ is an optimal control for the zero-sum game value $\VV$ in \eqref{eq:Def_FRT_valueFunction}.

\subsection{Proof of Theorem \ref{thm:cbf_is_frt}}
\label{appendix:cbf_is_frt}

We will show the two statements as follows.
\begin{enumerate}
    \item $\forall v\in C^\infty(\R^n)$ such that $\VV-v$ has a local minimum at $x_0\in\R^n$ and $\VV(x_0)=v(x_0)$,
    \vspace{-0.5em}
\begin{equation*}
    \!\!\!\!\!\!\!\!\!\!\!\!\!\resizebox{\hsize}{!}{$\displaystyle 0\leq \min\Big\{ v(x_0)-\cbf(x_0), \max_{u\in U}\min_{d\in D} \frac{\partial v}{\partial x}(x_0)\cdot f(x_0,u,d) + \gamma v(x_0)\Big\}.$} \vspace{-0.5em}
\end{equation*}
\item $\forall v\in C^\infty(\R^n)$ such that $\VV-v$ has a local maximum at $x_0\in\R^n$ and $\VV(x_0)=v(x_0)$, \vspace{-0.5em}
\begin{equation*}
    \!\!\!\!\!\!\!\!\!\!\!\!\!\resizebox{\hsize}{!}{$\displaystyle 0\geq \min\Big\{v(x_0)-\cbf(x_0), \max_{u\in U}\min_{d\in D} \frac{\partial v}{\partial x}(x_0)\cdot f(x_0,u,d) + \gamma v(x_0)\Big\}.$} \vspace{-0.5em}
\end{equation*}
\end{enumerate}
    
    \textit{Case 1.} $\VV(x_0)\!=\!\cbf(x_0)\!>\!0$: By the continuity of $\cbf$, there exists $\epsilon>0$ such that $\VV(y)\!=\!\cbf(y)$ for all $y\!\in\!B_\epsilon(x_0)$. Thus, the gradient of $\VV$ at $x_0$ exists: $\frac{\partial \VV}{\partial x}(x_0)=\frac{\partial \cbf}{\partial x}(x_0)$, so for any $v$ such that $\VV-v$ has either a local minimum or a local maximum at $x_0$, $\frac{\partial v}{\partial x}(x_0)=\frac{\partial \cbf}{\partial x}(x_0)$. From \eqref{eq:CBF_ineqality}, \vspace{-0.5em}
    \begin{equation} \label{eq:cbf-equality} 
        \max_{u\in U}\min_{d\in D} \frac{\partial v}{\partial x}(x_0)\cdot f(x_0,u,d) + \gamma v(x_0) =0.
    \vspace{-0.5em}
    \end{equation}
    Therefore, Statements 1) and 2) hold in this case.
    
    \textit{Case 2.} $\VV(x_0)\!=\!0\!>\!\cbf(x_0)$: By the continuity of $\cbf$, there exists $\epsilon>0$ such that $\VV(y)=0$ for all $y\in B_\epsilon(x_0)$. This implies that the gradient of $\VV$ at $x_0$ is $0\in\R^n$, so for any $v$ such that $\VV-v$ has either a local minimum or a local maximum at $x_0$, $\frac{\partial v}{\partial x}(x_0)=0$. Thus, \eqref{eq:cbf-equality} holds.
    Therefore, Statements 1) and 2) hold in this case.
    
    \textit{Case 3.} $\VV(x_0)\!=\!0\!=\!\cbf(x_0)$: From $v(x_0) - \cbf(x_0)\!=\!0$, it is trivial that 2) holds, and we focus on the proof of 1). Since $\VV-v$ has a local minimum at $x_0$, 
    $\frac{\partial v}{\partial x}(x_0)\in \partial^- \VV(x_0)$,
    where $\partial^- \VV(x_0)$ is the sub-differential, which is determined as 
        $\partial^- \VV(x_0) = \text{conv}\left(\{0\}\cup\Big\{\frac{\partial \cbf}{\partial x}(x_0)\Big\}\right)$,
    where $\text{conv}$ is a convex-hull operator. Thus, 
    $\frac{\partial v}{\partial x}(x_0) = \alpha \frac{\partial \cbf}{\partial x}(x_0)$
    for some $\alpha\in[0,1]$. Thus, from \eqref{eq:CBF_ineqality} and $v(x_0) = 0$,
    \eqref{eq:cbf-equality} holds and therefore, 1) holds.

\section*{Acknowledgments}
We thank Professors Helena Frankowska and Lawrence C. Evans for insightful discussions.

\balance

\bibliographystyle{plain}        
\bibliography{references}  
\end{document}

%% file: packages.tex
\usepackage{amsopn}

\usepackage{lipsum}
\usepackage{amsfonts}
\usepackage{graphicx}
\usepackage{epstopdf}
\usepackage{algorithmic}
\usepackage{xcolor}
\ifpdf
  \DeclareGraphicsExtensions{.eps,.pdf,.png,.jpg}
\else
  \DeclareGraphicsExtensions{.eps}
\fi

\usepackage{amsmath} 
\let\theoremstyle\undefined
\usepackage{amsthm}
\usepackage{amsfonts}
\usepackage{amssymb}  
\usepackage{mathtools} 
\usepackage{mathrsfs}
\usepackage{balance}
\DeclareMathAlphabet{\mathpzc}{OT1}{pzc}{m}{it}

\usepackage{url}

\usepackage[utf8]{inputenc}
\usepackage[T1]{fontenc}
\usepackage[sort,compress]{cite}

\usepackage{hyperref}

\usepackage{multirow}
\usepackage{multicol}

\usepackage{enumitem}

\usepackage{caption}
\usepackage{arydshln}

\definecolor{mymagenta}{rgb}{0.937, 0.004, 0.584}
\definecolor{myorange}{rgb}{0.965, 0.529, 0.255}
\definecolor{myblue}{rgb}{0.106, 0.588, 0.953}
\definecolor{mypurple}{rgb}{0.627, 0.125, 0.941}

\newtheoremstyle{spaced}  
  {0.5em}                  
  {0.5em}                  
  {\normalfont}              
  {}                      
  {\bfseries}             
  {.}                     
  {1em}                   
  {}                      
\theoremstyle{spaced} 

%% file: notations.tex
\newcommand{\R}{\mathbb{R}}

\newtheorem{proposition}{Proposition}
\newtheorem{definition}{Definition}
\newtheorem{corollary}{Corollary}
\newtheorem{lemma}{Lemma}
\newtheorem{theorem}{Theorem} 
\newtheorem{remark}{Remark}
\newtheorem{assumption}{Assumption} 
\newtheorem{condition}{Condition}

\newcommand{\interior}[1]{\mathrm{Int}(#1)}

\newcommand{\norm}[1]{\left\lVert #1\right\rVert}

\newcommand{\arxiv}[1]{\href{http://arxiv.org/abs/#1}{\normalsize{\textsc{arxiv:}}~\nolinkurl{#1}}}

\usepackage{color}
\usepackage{comment}

\renewcommand{\epsilon}{\varepsilon}
\renewcommand{\phi}{\varphi}

\newcommand{\N}{\mathbb{N}}

\newcommand{\X}{X}
\newcommand{\traj}{\mathrm{x}}
\newcommand{\BWtraj}{\mathrm{x}^-}

\newcommand{\finalrev}[1]{{\textcolor{black}{#1}}}

\newcommand{\dLee}[1]{{\color{black} #1}}
\newcommand{\shnote}[1]%
    {\textcolor{red}{ #1}}

\newcommand{\cbf}{h_S}
\newcommand{\dcbf}{\frac{\partial \cbf}{\partial x}}
\newcommand{\strategy}{\xi_{d}}
\newcommand{\Strategy}{\Xi_{d}}
\newcommand{\BWstrategy}{\xi_{d}^-}
\newcommand{\BWStrategy}{\Xi_{d}^-}
\newcommand{\UU}{\mathcal{U}}

\newcommand{\ctrl}{\mathrm{u}}
\newcommand{\BWctrl}{\mathrm{u}^-}
\newcommand{\dstb}{\mathrm{d}}
\newcommand{\BWdstb}{\mathrm{d}^-}

\newcommand{\cset}{U}
\newcommand{\cfset}{\UU}
\newcommand{\BWcfset}{\UU^-}
\newcommand{\dset}{D}
\newcommand{\dfset}{\mathcal{D}}
\newcommand{\BWdfset}{\mathcal{D}^-}

\newcommand{\cbfconstraint}{barrier constraint}

\newcommand{\discount}{\gamma}
\newcommand{\VV}{V_{\discount}}
\newcommand{\JJ}{J_{\discount}}

\newcommand{\frt}[1]{\mathrm{FRT}(#1)}
\newcommand{\minfrt}[1]{\mathrm{FRT}(#1)}

\newcommand{\vk}[1]{\mathrm{VK}^{-}(#1)}

\newcommand{\brt}[1]{\mathrm{BRT}(#1)}

\makeatletter
\newcommand{\newparallel}{\mathrel{\mathpalette\new@parallel\relax}}
\newcommand{\new@parallel}[2]{%
  \begingroup
  \sbox\z@{$#1T$}
  \resizebox{!}{\ht\z@}{\raisebox{\depth}{$\m@th#1/\mkern-5mu/$}}%
  \endgroup
}
\makeatother

%% file: table_comparison.tex
\begin{table*}
\centering
\caption{Comparison of reachability methods discussed in the paper. Background in Section \ref{sec:probDef} and discussion in Section \ref{sec:discussion}.}
\vspace{0.5em}
\renewcommand{\arraystretch}{1.5} 
\addtolength{\tabcolsep}{-4pt}
\begin{tabular}{|l||l|c|c|c|c|}
\hline
\multirow{2}{*}{{\footnotesize Methods}} & {\footnotesize Value function} & \multirow{2}{*}{\begin{tabular}{c} {\footnotesize Matching}  \\ {\footnotesize barrier constraint} \end{tabular}} & \multirow{2}{*}{{\footnotesize Boundedness}} & \multirow{2}{*}{{\footnotesize Continuity}} & \multirow{2}{*}{\begin{tabular}{c} {\footnotesize Sol. Unique.} \\ {\footnotesize of HJ-PDE} \end{tabular}} \\ \cdashline{2-2}
& {\footnotesize Diff. Inequality in HJ-PDE} &                    &                    &                    &                    \\ \hline \hline
\multirow{2}{*}{\footnotesize BRT with discount \cite{xue2018reach, akametalu2018minimum} } & \begin{tabular}{l}\quad\scriptsize $V(x) := \underset{\strategy}{\inf}\;\underset{\ctrl}{\sup}\underset{t\in[0,\infty)}{\inf} e^{-\discount t} h_{S}(\traj(t))$\end{tabular}& \multirow{2}{*}{no} & \multirow{2}{*}{\textbf{yes}} & \multirow{2}{*}{\textbf{yes}} & \multirow{2}{*}{\textbf{yes}} \\ \cdashline{2-2}
& \begin{tabular}{c} {\vspace{0.25em}\footnotesize \quad $\underset{u \in U}{\max}\underset{d \in D}{\min}\frac{\partial V}{\partial x}\!\cdot\!f(x,u,d)\!-\!\discount V \!\geq\!0$}\end{tabular} &                    &                    &                    &                    \\ \hline
\multirow{2}{*}{\footnotesize CBVF \cite{choi2021robust, tonkens2022refining} } & \begin{tabular}{l} \quad\scriptsize $V(x) := \underset{\strategy}{\inf}\;\underset{\ctrl}{\sup}\underset{t\in[0,\infty)}{\inf} e^{\discount t} h_{S}(\traj(t))$\end{tabular} & \multirow{2}{*}{\textbf{yes}} & \multirow{2}{*}{no} & \multirow{2}{*}{no} & \multirow{2}{*}{no} \\ \cdashline{2-2}
& \begin{tabular}{c} {\vspace{0.25em}\footnotesize \quad $\underset{u\in U }{\max}\underset{d\in D}{\min}\frac{\partial V}{\partial x}\!\cdot\!f(x,u,d)\!+\!\discount V \! \geq \! 0$ }\end{tabular} &                    &                    &                    &                    \\ \hline
\multirow{2}{*}{\footnotesize BRT w/o discount \cite{Fisac2018}} & \begin{tabular}{l} \quad{\scriptsize  $V(x) := \underset{{\strategy}}{\inf}\;\underset{\ctrl}{\sup}\underset{t\in[0,\infty)}{\inf}  h_{S}(\traj(t))$} \end{tabular} & \multirow{2}{*}{no} & \multirow{2}{*}{\textbf{yes}} & \multirow{2}{*}{no} & \multirow{2}{*}{no} \\ \cdashline{2-2}
& \begin{tabular}{c} {\vspace{0.25em}\footnotesize \quad $\underset{u \in U}{\max}\underset{d\in D}{\min}\frac{\partial V}{\partial x}\cdot f(x,u,d) \geq 0$} \end{tabular} &                    &                    &                    &                    \\ \hline
\multirow{2}{*}{\begin{tabular}{l}{\footnotesize \textbf{FRT with discount}}\\{\footnotesize \textbf{(Ours)}}
\end{tabular}}
 & \begin{tabular}{l}\quad\scriptsize $V(x) := \underset{\strategy}{\sup}\;\underset{\ctrl}{\inf}\underset{t\in(-\infty,0]}{\sup} e^{\discount t} h_{S}(\traj(t))$ 
\end{tabular} & \multirow{2}{*}{\textbf{yes}} & \multirow{2}{*}{\textbf{yes}} & \multirow{2}{*}{\textbf{yes}} & \multirow{2}{*}{\textbf{yes}} \\ \cdashline{2-2}
& \begin{tabular}{c} {\vspace{0.25em}\footnotesize \quad $\underset{u\in U}{\max}\underset{d\in D}{\min}\frac{\partial V}{\partial x}\!\cdot\!f(x,u,d)\!+\!\discount V\!\geq \!0$ }\end{tabular} &                    &                    &                    &                    \\ \hline
\end{tabular}
\label{table:comparision_CBF_Reachability}
\end{table*}